\documentclass[prx,a4paper,aps,twocolumn,superscriptaddress,longbibliography]{revtex4-1}
\usepackage{amsmath, empheq,braket,amsthm,amssymb,graphics,amsfonts,array,color,subfigure}
\definecolor{myurlcolor}{rgb}{0,0,0.7}
\definecolor{myurlcolor1}{rgb}{0,0.7,0.1}
\definecolor{myrefcolor}{rgb}{0,0,0.7}
\usepackage[unicode=true,pdfusetitle, bookmarks=true,bookmarksnumbered=false,bookmarksopen=false, breaklinks=false,pdfborder={0 0 0},backref=false,colorlinks=true, linkcolor=myrefcolor,citecolor=myurlcolor1,urlcolor=myurlcolor]{hyperref}

\theoremstyle{plain}
\newtheorem{thm}{\protect\theoremname}
\newtheorem{prop}[thm]{Proposition}
\newtheorem{rem}[thm]{Remark}

\providecommand{\theoremname}{Theorem}
\newcommand*{\myproofname}{Proof}

\newcommand{\act}{A_l}
\newcommand{\fs}{\mathcal{I}_{\mathrm{f}}}

\newcommand{\fo}{\Lambda_{\mathrm{f}}}
\newcommand{\IW}{\mathcal{I}_{\mathrm{W}}}

\DeclarePairedDelimiterX{\infdivx}[2]{(}{)}{%
  #1\;\delimsize\|\;#2%
}
\newcommand{\infdiv}{S\infdivx}

\begin{document}

\title{Quantum thermodynamics in a multipartite setting: A resource theory of local Gaussian work extraction for multimode bosonic systems}
\author{Uttam Singh}
\email{utsingh@ulb.ac.be}
\affiliation{Centre  for Quantum Information and Communication, {\'E}cole polytechnique de Bruxelles,\\
CP 165, Universit{\'e} libre de Bruxelles, 1050 Bruxelles, Belgium}
\author{Michael G. Jabbour}
\email{mgjabbour@maths.cam.ac.uk}
\affiliation{Centre  for Quantum Information and Communication, {\'E}cole polytechnique de Bruxelles,\\
CP 165, Universit{\'e} libre de Bruxelles, 1050 Bruxelles, Belgium}
\affiliation{DAMTP, Centre for Mathematical Sciences, University of Cambridge, Cambridge CB3 0WA, UK}
\author{Zacharie Van Herstraeten}
\email{zvherstr@ulb.ac.be}
\author{Nicolas J. Cerf}
\email{ncerf@ulb.ac.be}
\affiliation{Centre  for Quantum Information and Communication, {\'E}cole polytechnique de Bruxelles,\\
CP 165, Universit{\'e} libre de Bruxelles, 1050 Bruxelles, Belgium}

\begin{abstract}
Quantum thermodynamics can be cast as a resource theory by considering free access to a heat bath, thereby viewing the Gibbs state at a fixed temperature as a free state and hence any other state as a resource. Here, we consider a multipartite scenario where several parties attempt at extracting work locally, each having access to a local heat bath (possibly with a different temperature), assisted with an energy-preserving global unitary. As a specific model, we analyze a collection of harmonic oscillators or a multimode bosonic system. Focusing on the Gaussian paradigm, we construct a reasonable resource theory of local activity for a multimode bosonic system, where we identify as free any state that is obtained from a product of thermal states (possibly at different temperatures) acted upon by any linear-optics (passive Gaussian) transformation. The associated free operations are then all linear-optics transformations supplemented with tensoring and partial tracing. We show that the local Gaussian extractable work (if each party applies a Gaussian unitary, assisted with linear optics) is zero if and only if the covariance matrix of the system is that of a free state. Further, we develop a resource theory of local Gaussian extractable work, defined as the difference between the trace and symplectic trace of the covariance matrix of the system. We prove that it is a resource monotone that cannot increase under free operations. We also provide examples illustrating the distillation of local activity and local Gaussian extractable work.
\end{abstract}

\maketitle

\section{Introduction}
Thermodynamics is a macroscopic theory applicable in the limit where the number of particles and volume tend to infinity \cite{Callen1985}.  However, with our increasing ability to control or manipulate small systems and the realization of molecular motors \cite{Rousselet1994, Howard1997, Hanggi2009} and nanomachines \cite{Scovil1959, Geusic1967, Alicki1979, Scully2002, Faucheux1995}, the scope of applicability of thermodynamics is starting to stretch beyond the macroscopic region. One of the main goals of the thermodynamics of small systems--quantum thermodynamics--is the extraction of work by means of cyclic Hamiltonian transformations of a quantum state. Evidently, it is of great importance to know which states do not allow for any work extraction under Hamiltonian transformations. Such states are known as {\it passive states} \cite{Pusz1978, Lenard78}. For a quantum system in a state $\rho$ with a given Hamiltonian $\hat{H}$, the maximum amount of work that can be extracted using any unitary $U$ is defined as 
\begin{align*}
W(\rho,\hat{H}):=\max_{U}\mathrm{Tr}[\hat{H}\left(\rho - U\rho U^\dagger\right)].
\end{align*}
Thus, a passive state $\rho_{\mathrm{p}}$ is such that  $W(\rho_{\mathrm{p}},\hat{H})=0$. It is also known that, given a passive state $\rho_{\mathrm{p}}$, a tensor power of it may or may not be passive; \textit{i.e.}, $W(\rho_{\mathrm{p}}^{\otimes n},\hat{H}_{\mathrm{tot}})$ may or may not be zero for some integer $n$, where $\hat{H}_{\mathrm{tot}}$ is the total Hamiltonian. A passive state $\rho_{\mathrm{p}}$ that remains passive for all its tensor powers $\rho_{\mathrm{p}}^{\otimes n}$, $\forall n$, is known as completely passive. A central result in quantum thermodynamics is that the only completely passive states are the thermal Gibbs states $\rho \propto \exp(-\beta \hat H)$ \cite{Lenard78}.

A resource theory of thermodynamics can be developed to systematically describe work extraction from a quantum system and, in general, the allowed state transformations are those where the system interacts via an energy-preserving unitary together with an ancilla chosen to be in a thermal Gibbs state (with an arbitrary Hamiltonian) at some fixed temperature \cite{Michal2013, Brandao2013, Brandao2015b, Goold2016}. In this resource-theoretic treatment of quantum thermodynamics, the thermal Gibbs state of the system at the same temperature as that of the ancilla is the only free state \cite{Michal2013, Brandao2013, Brandao2015b}. Although considering arbitrary Hamiltonians and arbitrary energy-preserving unitaries is satisfying in the context of establishing a general framework for quantum thermodynamics, it may also be interesting to focus on states and unitaries of higher practical relevance. For bosonic systems, for example, restricting to Gaussian states and Gaussian operations has proven to be very fruitful, particularly in the field of quantum information theory with continuous variables \cite{Weedbrook2012, Olivares2012, Adesso2014}. Similarly, exploring quantum thermodynamics with Gaussian bosonic systems is a promising avenue \cite{Brown2016}, which we investigate here.

\begin{figure}
\centering
\includegraphics[width=65 mm]{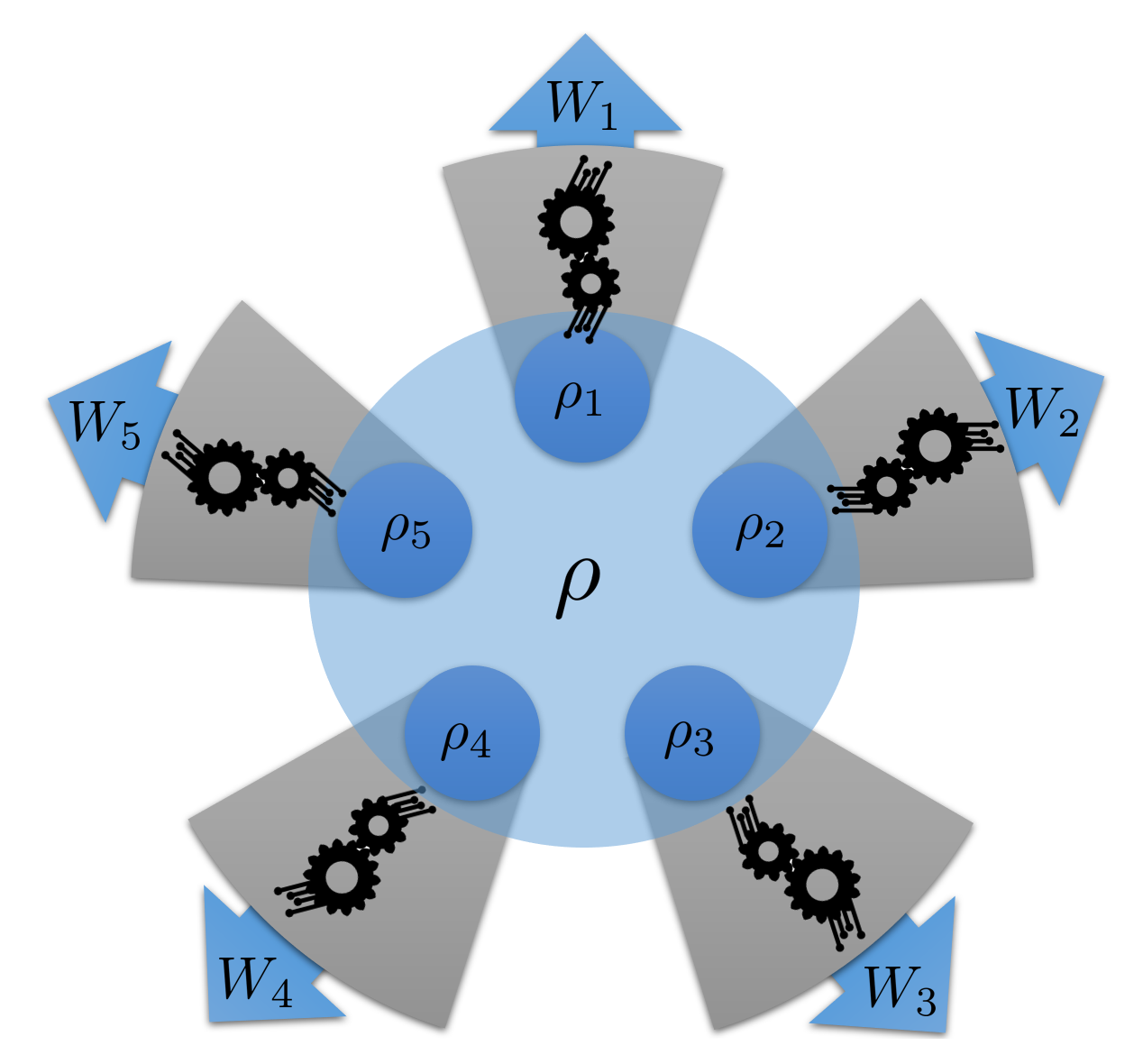}
\caption{Multipartite quantum thermodynamical scenario in which each party extracts work locally and this process is being assisted with a global energy-preserving unitary (allowing the parties to exchange energy among them but not allowing global work extraction). Here, we consider the model where the quantum system held by each party is a harmonic oscillator (or a bosonic mode) and local work extraction is restricted to Gaussian unitaries (especially squeezing). It is assisted with a global energy-preserving (passive) Gaussian unitary, which corresponds to any linear-optics circuit.}
\label{fig:scenario}
\end{figure}


In this paper, we explore a {\it multipartite} quantum thermodynamical scenario as illustrated in Fig. \ref{fig:scenario}, where each party can extract work locally by applying a local unitary and this process is being assisted with a global energy-preserving unitary (hence, allowing no global work extraction as such). This is not a trivial extension of work extraction for a single party because there exist situations where an energy-preserving coupling allows the parties to extract work locally even though their local (reduced) states are initially passive. Given the definition of passive states, a natural choice may be to consider them (instead of Gibbs states) as free states in a resource theory for extractable work. However, considering passive states as free states defies a plausible criterion for any reasonable resource theory, namely that if a state $\rho$ is free, then $\rho^{\otimes n}$ should also be free for any integer $n$ \cite{Brandao2015}. For this reason, we rather take Gibbs states as building blocks of our free states for each party, which allows us to develop a multipartite resource theory for local extractable work within this restriction.
 

In particular, we examine a multimode bosonic system as a specific model, where the Hamiltonian is that of $N$ harmonic oscillators. Hence, the Gibbs states reduce to Gaussian thermal states and energy-preserving unitaries become passive Gaussian unitaries (\textit{i.e.}, all linear-optics transformations). In order to develop a multipartite resource theory, we consider as {\it free states} the products of thermal states (possibly at different temperatures) acted upon by linear-optics transformations (see Fig. \ref{fig:free-state}a). We first discuss the properties of this set of free states, denoted as $\fs$, and build the set of {\it free operations} $\fo$ that is consistent with $\fs$. Specifically, a free operation cannot create local activity (i.e., a resource state) when acting upon any free state. We note that $\fs$ is not convex, owing to the fact that the set of Gaussian states is not convex. This issue might be solved by using the convex hull of $\fs$ (see, e.g., \cite{Albarelli2018, Takagi2018}), but we choose not to follow this procedure here as there is a physically motivated way to define a convex set $\IW$ that contains $\fs$ (see Fig. \ref{fig:set-of-free-states}). Indeed, it appears that the covariance matrices of our free states, which we call free covariance matrices, form a convex set. Furthermore, we prove that the states admitting a free covariance matrix coincide with the states from which no work can be extracted by local Gaussian unitaries assisted with linear optics. This set of states with no extractable work, noted $\IW$, is therefore convex. Clearly, $\fs$ is contained in $\IW$ as it corresponds to the subset of Gaussian states within $\IW$, that is, the Gaussian states from which no local Gaussian work extraction assisted with linear optics is possible (the non-Gaussian states with no extractable work belong then to $\IW \setminus \fs$). It must be noticed that there exist states in $\IW$ that do not belong to the convex hull of $\fs$, as sketched in Fig. \ref{fig:set-of-free-states}. These are non-Gaussian states with no extractable work that cannot be written as convex mixtures of (Gaussian) free states.
 
\begin{figure}
\centering
\subfigure[~Free states.]{
\includegraphics[width=0.55\columnwidth]{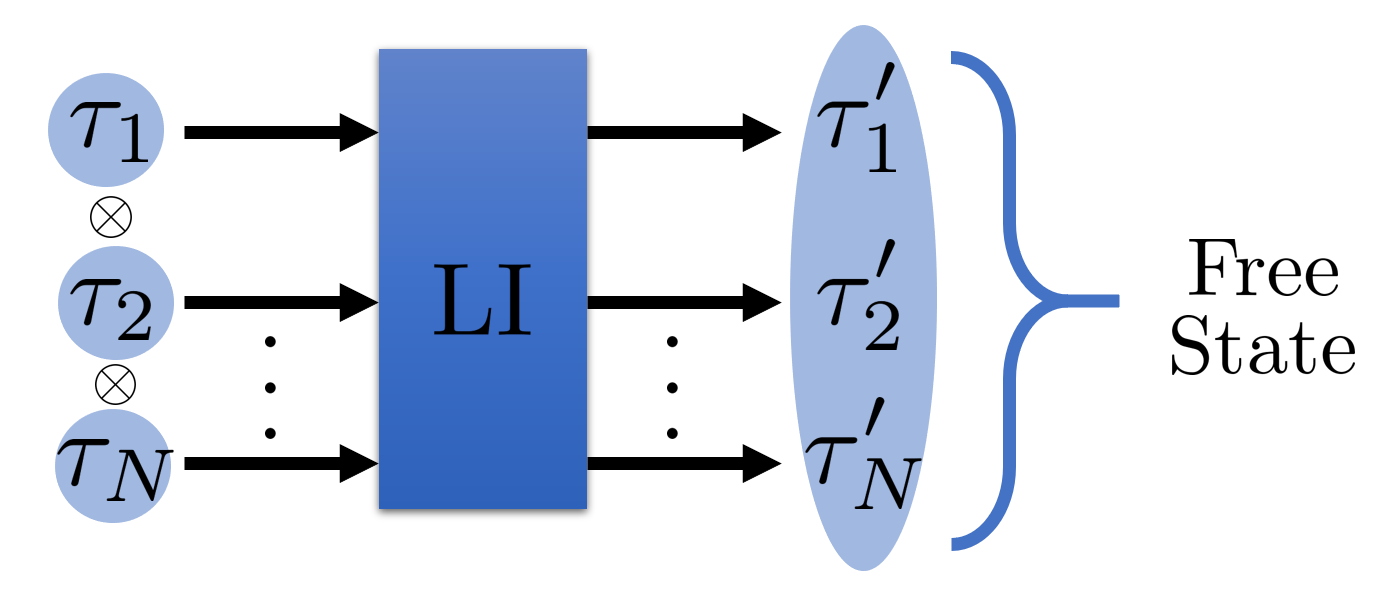}}
\subfigure[~Tensoring of free states.]{
\includegraphics[width=0.41\columnwidth]{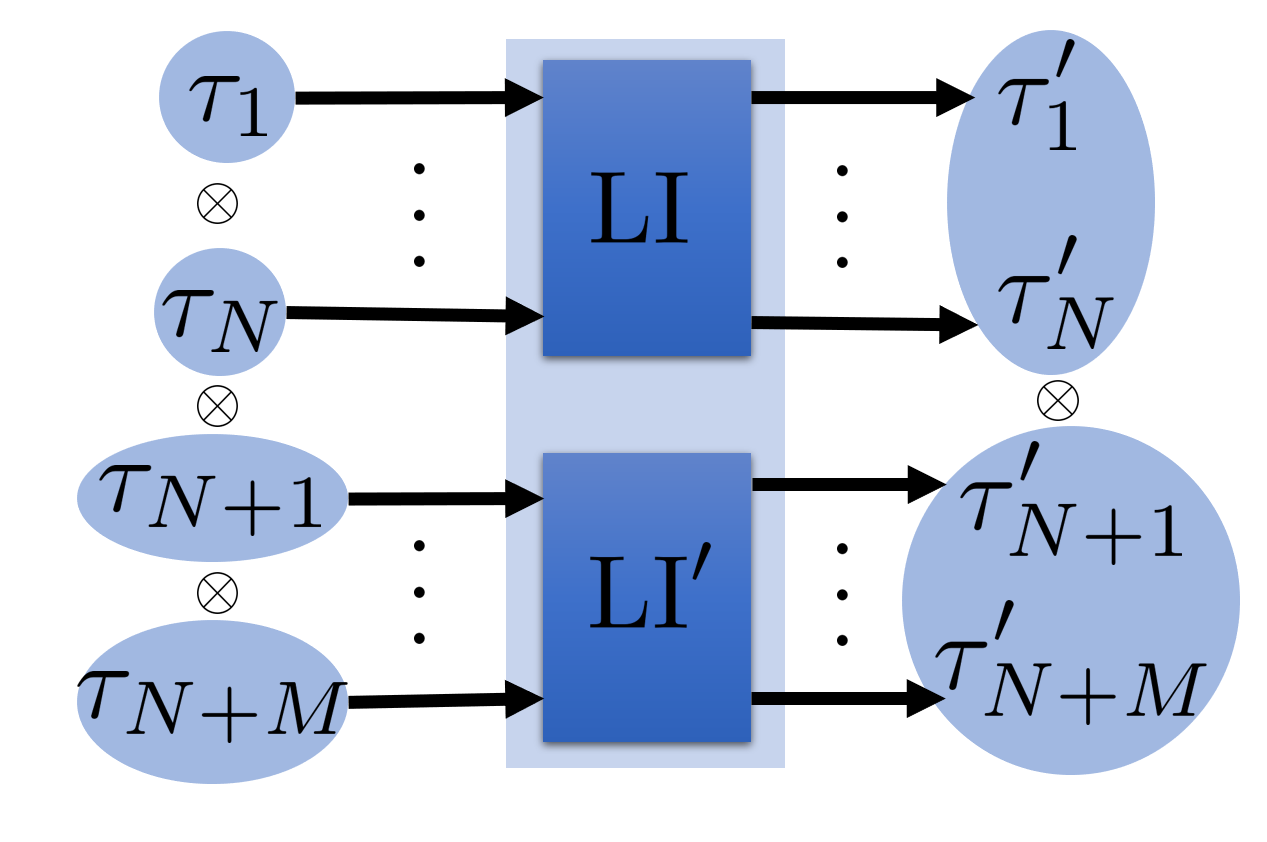}}
\caption{(a) Definition of an $N$-mode free state $\rho^f_{1\cdots N}$. Here, $\tau_1\otimes\cdots \otimes\tau_N$ is a tensor product of $N$ thermal states (possibly with different temperatures) and LI represents a linear interferometer. The single-mode reduced states of $\rho^f_{1\cdots N}$ are thermal states $\tau_1',\cdots, \tau_N'$. (b) From the definition (a), it is clear that the tensor product of free states is itself a free state. }
\label{fig:free-state}
\end{figure}

In this work, we start by developing a resource theory of {\it local activity} in terms of quantum states (that is, building on the set of free states $\fs$; hence not being in $\fs$ is viewed as a resource). We dub the resource states, \textit{i.e.}, the states $\rho \notin \fs$, as {\it locally active} states in the sense that they contrast with passive states. We develop resource monotones for local activity based on contractive distance measures, with a particular emphasis on relative entropy. We find that the relative entropy of local activity of a state $\rho$, noted $A_l(\rho)$,  is additive for product states; however, it is neither sub- nor superadditive for arbitrary quantum states (although we can express a relaxed form of subadditivity). Next, we explicitly calculate the relative entropy of local activity for arbitrary two-mode Gaussian states. Moreover, through explicit examples, we show that it is possible to obtain a more resourceful state starting from two copies of a less resourceful state. This shows that the distillation of local activity is, in principle, possible.

In the second part of this work, we develop a resource theory in terms of covariance matrices (that is, building on the set of free covariance matrices instead of free states; hence not being in $\IW$ is viewed here as a resource). This leads to a resource theory of {\it local Gaussian work extraction} assisted with linear optics. We elaborate on the properties of local Gaussian extractable work in this multipartite setting, noted $W_l(\rho)$, and discuss its possible distillation. 
Interestingly, it can be expressed as the difference between the trace and symplectic trace of the covariance matrix associated with the state, allowing us to access its properties by exploiting the symplectic formalism of quantum optics.





\begin{figure}
\centering
\includegraphics[width=0.65\columnwidth]{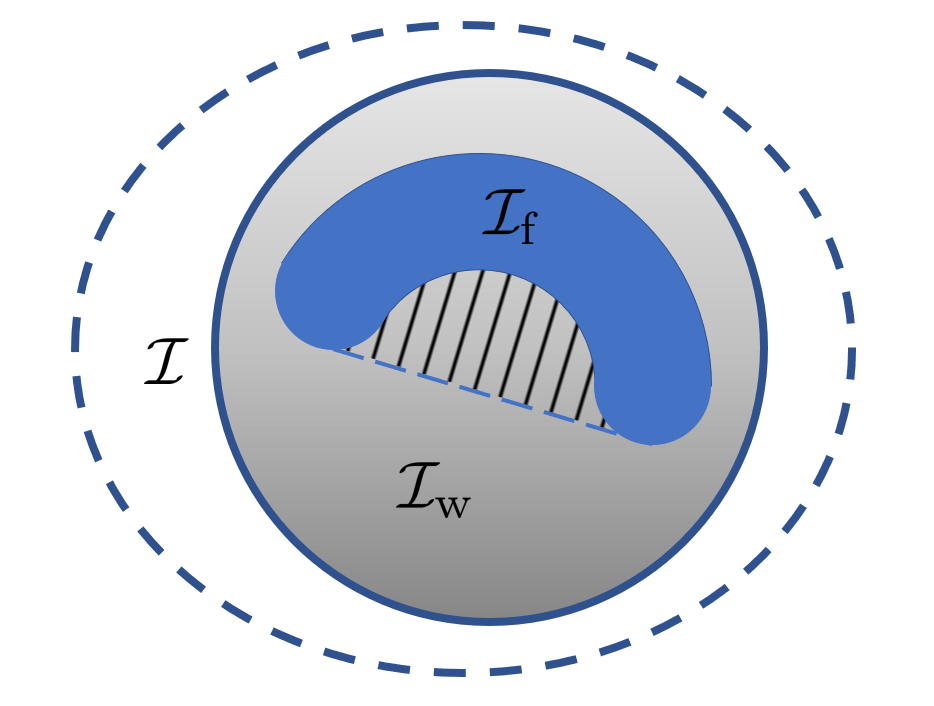}
\caption{Within the set of all quantum states $\cal I$, we define the set of states $\IW$ from which no local Gaussian work extraction (assisted with linear optics) is possible, namely states $\rho$ such that  $W_l(\rho)=0$.  The subset of Gaussian states within $\IW$ is our set of free states $\fs$, namely states $\rho$ whose local activity $A_l(\rho)=0$. Note that $\fs$ is not convex, owing to the fact that Gaussian states do not form a convex set, but it is strictly included in the convex set $\IW$. Thus, there exist non-Gaussian locally active states ($A_l\ne 0$) from which no local Gaussian extraction of work is possible ($W_l=0$).  }
\label{fig:set-of-free-states}
\end{figure}

The rest of the paper is organized as follows. We first set the notations and some preliminaries in  Sec. \ref{sec:prelims}. In Sec. \ref{sec:basic},  we introduce the set of free states and free operations for our resource theory of local activity and discuss their properties. In Sec. \ref{sec:monotones}, we define local activity monotones and then provide some explicit calculations in Sec. \ref{sec:mono-calc}. In particular, we derive a closed-form formula for the relative entropy of local activity for arbitrary two-mode Gaussian states.  In Sec. \ref{sec:extra-work}, we introduce local Gaussian extractable work viewed as a resource, and discuss its various properties. In Sec. \ref{sec:distillation}, we discuss the possibility of distillation of various resources. We conclude in Sec. \ref{sec:conclusion}, with a discussion on the implications of our findings. Finally, in the appendices,  we provide details on some of our calculations.


\section{Preliminaries and notations}
\label{sec:prelims}

\textbf{Gaussian states:} Let us consider a system of $N$ bosonic
modes with quadrature operators ${\bf\hat{x}}=(\hat{q}_{1},\hat{p}_1\dots,\hat{q}_{N}%
,\hat{p}_{N})^{T}$ which satisfy the canonical commutation relations \cite{Weedbrook2012}
\begin{align*}
\lbrack\hat{x}_i,\hat{x}_j] = \mathfrak{i} \, \Omega_{ij}~~~(i,j=1,\cdots,2N),
\end{align*}
where
\begin{align}
\mathbf{\Omega}=\bigoplus\limits_{k=1}^{N} \pmb{\omega},~~~~~~\pmb{\omega}=
\begin{pmatrix}
0 & 1\\
-1 & 0
\end{pmatrix},
\end{align}
and we have set $\hbar=1$.
The corresponding $N$ pairs of annihilation and creation operators are defined as $\hat{a}_i=\frac{1}{\sqrt{2}}(\hat{q}_i+\mathfrak{i} \hat{p}_i)$ and $\hat{a}_i^\dagger=\frac{1}{\sqrt{2}}(\hat{q}_i-\mathfrak{i} \hat{p}_i)$. The Hamiltonian corresponding to mode $i$ is given by $\hat{H}_i=\hat{a}_i^\dagger \hat{a}_i +1/2$, where we have considered all angular frequencies to be the same and equal to one. Given an $N$-mode quantum state $\rho$, the first-order moments constitute the displacement vector, defined as
\begin{equation}
{\bf\bar{x}}:=\left\langle {\bf\hat{x}}\right\rangle =\mathrm{Tr}({\bf\hat{x}}\hat{\rho}).
\end{equation}
The second-order moments make the covariance matrices (CMs), defined as
\begin{equation}
\Gamma_{ij}:=\frac{1}{2}\left\langle \left\{  \hat{x}_{i}-\langle\hat{x}_{i}\rangle,\hat{x}_{j}-\langle\hat{x}_{j}\rangle\right\}  \right\rangle , \label{CM_definition}%
\end{equation}
where $\{ \bullet , \bullet \}$ represents the anticommutator. The matrix $\pmb{\Gamma}$ is a positive-definite matrix. In particular, any positive-definite matrix which satisfies the uncertainty relation qualifies as a valid CM \cite{Simon1994}. 

A Gaussian quantum state $\rho=\rho(\mathbf{\bar{x}},\pmb{\Gamma})$ has a Gaussian Wigner representation. As a consequence, it is described fully in terms of its first two statistical moments, namely, the displacement vector and the CM \cite{Holevo2011}. The vacuum state $\ket{0}$ is a Gaussian state with zero displacement and CM $\pmb{\Gamma}=\frac{1}{2}\mathbb{I}_2$, where $\mathbb{I}_2$ denotes a $2\times 2$ identity matrix. Similarly, the thermal state $\rho_{\mathrm{th}}$ is a Gaussian state with zero displacement and CM $\pmb{\Gamma}=(\bar{n}+\frac{1}{2})\mathbb{I}_2$, where $\bar{n}= \mathrm{Tr}[\rho_{\mathrm{th}}\hat{a}^\dagger \hat{a}]$ \cite{Weedbrook2012}.

The total energy of a system of $N$ bosonic
modes in an arbitrary state $\rho$ is given by $E=\sum_i \mathrm{Tr}[\rho \hat{H}_i]$. This can be reexpressed as 
\begin{align}
\label{eq:ene-exp}
E=\frac{1}{2} \left(\mathrm{Tr}[\pmb\Gamma] +  |\bf\bar{x}|^2\right),
\end{align}
where $\pmb\Gamma$ and $\bf\bar{x}$ are the CM and displacement vector of state $\rho$, respectively. Note that the expression of the energy holds for any state (Gaussian or not).

It is worthwhile to notice that an arbitrary $N$-mode Gaussian state $\rho(\mathbf{\bar{x}},\pmb{\Gamma})$ can be written as \cite{Banchi2015}
\begin{equation}
\rho(\mathbf{\bar{x}},\pmb{\Gamma})=\frac{\exp\left[  -\frac{1}{2}({\bf\hat{x}}-{\bf\bar{x}})^{T}{\bf G}({\bf\hat{x}}-{\bf\bar{x}})\right]  }%
{\det\left( \pmb \Gamma+\mathfrak{i}\pmb\Omega/2\right)  ^{1/2}},
\end{equation}
where the matrix ${\bf G}$ can be defined in terms of the covariance matrix $\pmb\Gamma$ as
\begin{equation}
\label{def:G}
{\bf G}=2\mathfrak{i}\pmb\Omega\,\coth^{-1}(2\pmb\Gamma \mathfrak{ i}\pmb\Omega).
\end{equation}
From the Williamson theorem \cite{Williamson1936, Weedbrook2012}, any CM $\pmb\Gamma$ can be brought into the form ${\bf D}=\bigoplus\nolimits_{k=1}^{N}\nu_{k}\mathbb{I}_{2}$ through expression $\pmb\Gamma=\mathcal{S}{\bf D}\mathcal{S}^T$, where $\mathcal{S}$ is a symplectic matrix (it satisfies $\mathcal{S}\pmb\Omega\mathcal{S}^T=\pmb\Omega$) and the variables $\nu_{k}$ are called symplectic eigenvalues (they satisfy the uncertainty principle $\nu_{k}\geq1/2$, $\forall k$). Using this, one has ${\bf G}=-\pmb\Omega \mathcal{S} \mathcal{G}({\bf D}) \mathcal{S}^T \pmb\Omega$, where $\mathcal{G}(x)= 2\coth^{-1}(2x)$ \cite{Banchi2015}.

\bigskip
\textbf{Gaussian unitary operations:} A Gaussian unitary transformation is a unitary transformation that preserves the Gaussian character of a quantum state \cite{Weedbrook2012}. In terms of quadrature operators, a Gaussian unitary transformation is an affine map
\begin{align}
(\mathcal{S}, {\bf d}):{\bf \hat{x}}\rightarrow \mathcal{S}{\bf\hat{x}}+{\bf d},
\end{align}
where $\mathcal{S}$ is a $2N\times 2N$ real symplectic matrix and ${\bf d}$ is a $2N\times 1$ real vector. Under Gaussian unitary transformations, ${\bf \bar{x}}\rightarrow\mathcal{S}{\bf \bar{x}}+{\bf d}$ and $\pmb\Gamma\rightarrow\mathcal{S}\pmb\Gamma\mathcal{S}^T$. The Gaussian unitary is called passive if it is energy conserving (or photon-number conserving). In the rest of this work, we denote passive Gaussian unitaries as $U^{\mathrm{PG}}$. Such a unitary implies an orthogonal symplectic transformation $\mathcal{S}$ on the quadrature operators. Physically, passive Gaussian unitaries correspond to all linear-optics circuits, that is, any multiport interferometer made of beam splitters and phase shifters.

\bigskip
\textbf{von Neumann entropy of a Gaussian state:}
The von Neumann entropy of a Gaussian state $\rho$ can be written as \cite{Holevo1999, Weedbrook2012}
\begin{align}
S(\rho)=\sum_{k=1}^{N}g(\nu_{k}),
\end{align}
where the $\nu_{k}$ are the symplectic eigenvalues of $\rho$, while
\begin{align}
g(y)=\left(y+\frac{1}{2}\right)\ln\left(y+\frac{1}{2}\right) - \left(y-\frac{1}{2}\right)\ln\left(y-\frac{1}{2}\right).
\end{align}
Here, the logarithm is considered in base $e$.
For a thermal state $\rho_{th}$, $\nu_{k}=\bar{n}_{k}+1/2$, where $\bar{n}_{k}= \mathrm{Tr}[\rho\hat{a}_k^\dagger \hat{a}_k]$ is the mean number of photons in the $k$th mode, so that its entropy is given by
\begin{align}
S(\rho_{th})=\sum_{k=1}^{N}\left[(\bar{n}_{k}+1)\ln(\bar{n}_{k}+1) - \bar{n}_{k}\ln\bar{n}_{k}\right].
\end{align}

\bigskip
\textbf{Relative entropy between Gaussian states:} The relative entropy between two arbitrary Gaussian states
$\rho({\bf \bar{x}}_{1},\pmb\Gamma_{1})$ and $\sigma({\bf \bar{x}}_{2},\pmb\Gamma_{2})$ is given by \cite{Pirandola2017}
\begin{align}
&S(\rho({\bf\bar{x}}_{1},\pmb\Gamma_{1})||\sigma({\bf\bar{x}}_{2},\pmb\Gamma_{2}))\nonumber\\
&=S(\rho)+ \frac{1}{2}\left[\ln\det\left(  \pmb\Gamma_{2}+\frac{i\pmb\Omega}%
{2}\right)  +\mathrm{Tr}(\pmb\Gamma_{1}{\bf G}_{2})+\pmb\delta^{T}{\bf G}_{2}\pmb\delta\right],
\label{functional}%
\end{align}
where $\pmb\delta:={\bf \bar{x}}_{1}-{\bf\bar{x}}_{2}$ and ${\bf G}_2$ is defined through Eq. (\ref{def:G}). See Refs. \cite{Scheel2001, Chen2005} for other formulations of the relative entropy between Gaussian states.

\section{Basic framework: Free states and free operations}
\label{sec:basic}
A general resource theory comprises two basic elements: the set of free states and the set of free operations. Based on these two elements, the resource states can be identified and the amount of the resource is then quantified with the help of resource monotones which satisfy certain {\it bonafide} criteria. For more details on the structure of resource theories, see Refs. \cite{Brandao2015, Lami2018}. For examples of well-studied resource theories, see Refs. \cite{Bartlett2007, Genoni2008, Horodecki2009, Michal2013, Brandao2013, Marvian2013, Marian2013, Marvian2014, Baumgratz2014, Streltsov2015, Brandao2015b, Streltsov2017, Gallego2015, Albarelli2018, Takagi2018, Zhuang2018, Chitambar2019, Contreras2019, Taddei2019}.

In the following, we introduce the free states and free operations suitable for our purposes of describing local Gaussian work extraction scenarios and discuss the implications of these two basic elements.

\bigskip
\noindent
{\bf(1) Free states:} A state is free if it is a passive Gaussian unitary equivalent of a product of thermal states, as illustrated in Fig.~\ref{fig:free-state}a. In other words, it is a product of thermal states acted upon by any passive Gaussian unitary. Let us denote the set of free states as $\fs = \{U^{\mathrm{PG}}(\tau_1\otimes\cdots\otimes \tau_N)U^{\mathrm{PG}\dagger}\}$, where $U^{\mathrm{PG}}$ is a passive Gaussian unitary transformation and $\{\tau_i\}$ are thermal states corresponding to different modes (possibly with different temperatures). As already mentioned, any passive Gaussian unitary can be built with linear optics (it is a concatenation of beam splitters and phase shifters). In the case of two modes, $U^{\mathrm{PG}}$ is a combination of just a single beam splitter and three phase shifters. In particular, any $N$-mode passive Gaussian unitary can be written as a concatenation of $N(N-1)/2$ beam splitters and $N(N+1)/2$ phase shifters \cite{Simon1994}.

\begin{rem}
The free states are Gaussian states such that the reduced state of each mode is a thermal state (the converse is not true). This directly follows from the structure of their covariance matrix [see Eq. \eqref{eq:free-cov}].
\end{rem}
\begin{rem}
The free states 
are separable. This follows from Refs. \cite{Kim2002, Xiang2002}, which state that the output state of a beam splitter is always a separable state if the input is classical; \textit{i.e.}, it has a positive $P$ function \cite{Sudarshan1963, Glauber1963, Cahill1969}. However, the converse is not true and all Gaussian separable states do not belong to the set of free states. In order to see that, consider a coherent state as an example of a state with a positive $P$ function. If a coherent state is fed in one input of a beam splitter and a vacuum in its other input, the output gives rise to a Gaussian separable state, while it is not free according to our definition.
\end{rem}
\begin{rem}
The set of free states $\fs$ is not convex. This follows from the fact that the set of Gaussian states is not convex; \textit{i.e.}, if we consider the convex combination of free states, then the resulting state will generally not be free as it might be non-Gaussian.
\end{rem}

Since the free states are Gaussian states, we can describe them using their displacement vector and covariance matrix. Indeed, an $N$-mode free state $U^{\mathrm{PG}}(\tau_1\otimes\cdots\otimes\tau_N)U^{\mathrm{PG}\dagger}$ can be seen as a Gaussian state with displacement zero and CM
 \begin{align}
 \label{eq:free-cov-unsimp}
 \pmb\Gamma(N)=\mathcal{O}\left(\oplus_{i=1}^{N} \nu_i \, \mathbb{I}_2\right) \mathcal{O}^T,
 \end{align}
 where $\nu_i=\bar{n}_i +1/2$, with $\bar{n}_i$ being the average photon number in the $i$th thermal state $\tau_i$ and $\mathcal{O}$ being an orthogonal symplectic transformation that corresponds to $U^{\mathrm{PG}}$. Furthermore, we show that the CM of an $N$-mode free state can be written in a simple form as 
\begin{align}
 \label{eq:free-cov}
 \pmb\Gamma(N)
 =\begin{pmatrix}
 a_{11}\mathbb{I}_2 &R_{12} &\cdots & R_{1N} \\
 R_{12}^T &a_{22} \mathbb{I}_2&\cdots & R_{2N} \\
 \vdots &\vdots &\ddots &\vdots\\
 R_{1N}^T  & R_{2N}^T&\cdots & a_{NN} \mathbb{I}_2
 \end{pmatrix},
 \end{align}
 where $R_{ij}$ are $2\times 2$ matrices such that $R_{ij}R_{ij}^T\propto \mathbb{I}_2$ and $R_{ij}\pmb{\omega} R_{ij}^T\propto \pmb{\omega}$. This can be proved using mathematical induction and the fact that any passive Gaussian unitary can be built with linear optics. The interested reader is referred to Appendix \ref{append:free-states} for a proof. Note also that the symplectic eigenvalues of $\pmb\Gamma(N)$ in Eq. \eqref{eq:free-cov} are the same as its eigenvalues.

\begin{rem}
The covariance matrices corresponding to free states form a convex set. That is, if $\pmb\Gamma_{free}^{(1)}$ and $\pmb\Gamma_{free}^{(2)}$ are two CMs corresponding to two free states (Eq. \eqref{eq:free-cov}), then for $0\leq p \leq 1$, $p \pmb\Gamma_{free}^{(1)} + (1-p) \pmb\Gamma_{free}^{(2)}$ is also of the form given by Eq. \eqref{eq:free-cov}. See Appendix \ref{append:non-convex} for a proof.  Note that this is not in contradiction with the fact that the set of free states is not convex. Indeed, the covariance matrix that we obtain after mixing two covariance matrices corresponding to Gaussian free states can very well describe a non-Gaussian state, which is not free according to our definition. However, in situations where we are only concerned with covariance matrices (for example, if we are interested in work extraction), we recover convexity of the set $\IW$ (see Fig. \ref{fig:set-of-free-states}) since the mixture of free covariance matrices is free.
\end{rem}

From the structure of the covariance matrices of our free states, Eq. \eqref{eq:free-cov}, it is obvious that locally, for each mode, we get the covariance matrix of a thermal state. Therefore, one understands that no work can be extracted locally from these free states. This hints at the fact that we are in the right direction if our goal is to develop a resource theory of local Gaussian work extraction. In fact, we will later prove that the free covariance matrices are actually the only covariance matrices that do not allow for any local Gaussian work extraction, even when assisted with linear optics.

As already noted, free covariance matrices may also characterize non-Gaussian states in state space. As a consequence, defining the resource of an arbitrary state in terms of its distance to the set of free states $\fs$ in state space does not necessarily quantify its usefulness for local Gaussian work extraction. Therefore, we choose the denomination {\it local activity} for the distance-based monotones in state space as defined in Sec. \ref{sec:monotones}, making an explicit distinction with the {\it local extractable work} later considered in Sec. \ref{sec:extra-work} based on the phase-space picture. This distinction is connected to the fact there exist states out of $\fs$ (hence, locally active) belonging to $\IW$ (hence, having no local Gaussian extractable work) as shown in Fig. \ref{fig:set-of-free-states}.
\\ \ \\

\noindent
 {\bf(2) Free operations:} 
 
 \bigskip
 \noindent
 {\it $(O_1)$ Passive (energy-conserving) Gaussian unitaries:--} These are by definition free operations.
 
 \bigskip
 \noindent
  {\it $(O_2)$ Tensoring of free states:--} Given an $N$-mode free state $U^{\mathrm{PG}}(\tau_1\otimes\cdots\otimes\tau_N)U^{\mathrm{PG}\dagger}$, if we tensor it with any other $M$-mode free state $V^{\mathrm{PG}}(\tau_1'\otimes\cdots\otimes\tau_M')V^{\mathrm{PG}\dagger}$, then $U^{\mathrm{PG}}\otimes V^{\mathrm{PG}}(\tau_1\otimes\cdots\otimes\tau_N)\otimes  (\tau_1'\otimes\cdots\otimes\tau_M') U^{\mathrm{PG}\dagger}\otimes V^{\mathrm{PG}\dagger}$ is again a free state as $U^{\mathrm{PG}}\otimes V^{\mathrm{PG}}$ is another passive Gaussian unitary. Similarly, since tensoring in state space means applying direct summation in phase space, it is easy to understand that, if $\pmb\Gamma_{free}^{(1)}$ and $\pmb\Gamma_{free}^{(2)}$ are two free covariance matrices, then $\pmb\Gamma_{free}^{(1)}\oplus\pmb\Gamma_{free}^{(2)}$ is also free. This is obvious from Eq. (\ref{eq:free-cov}).

 \bigskip
 \noindent
 {\it $(O_3)$ Partial tracing:--} That partial tracing is a free operation can be seen most conveniently in the phase-space picture. Suppose we have an $N$-mode free state as given in Eq. \eqref{eq:free-cov-unsimp}.  Now, suppose we trace out one mode, say the last one; then, at the level of covariance matrices, this translates into deleting rows and columns of the corresponding mode. The remaining covariance matrix $ \pmb\Gamma'(N-1)$ then corresponds to the partial-traced state. Now, let us show that the remaining covariance matrix can be written as 
  \begin{align*}
 \pmb\Gamma'(N-1)=\mathcal{O}'\left(\oplus_{i=1}^{N-1} \nu'_i \, \mathbb{I}_2\right) \mathcal{O}^{'T} ,
 \end{align*}
 so that the remaining state is indeed free, as illustrated in Fig.~\ref{fig:tensoring-free-state}. To prove the above, it suffices to prove that $\pmb\Gamma'(N-1)$ has a similar structure to Eq. \eqref{eq:free-cov} and has all its eigenvalues greater than or equal to half. The fact that $\pmb\Gamma'(N-1)$ has eigenvalues greater than or equal to half follows from Cauchy's interlacing theorem for symmetric matrices. In particular, $\lambda_{\mathrm{min}}( \pmb\Gamma'(N-1)) \geq \lambda_{\mathrm{min}}( \pmb\Gamma_{N}) \geq 1/2$, where $\lambda_{\mathrm{min}}( \pmb\Gamma_{N})$ represents the smallest eigenvalue of $\pmb\Gamma_{N}$. 
%
 %
The fact that the reduced covariance matrix $\pmb\Gamma'(N-1)$ can be written in a similar form to that in Eq. \eqref{eq:free-cov}  is clear since  we simply have deleted the last two rows and columns.
 
Thus, having in mind that in state space, partial-tracing over a multimode Gaussian state yields another Gaussian state, one understands that the resulting $(N-1)$-mode state corresponds to a free state. Therefore, we conclude that partial-tracing one of the modes is a free operation. The same argument can be applied recursively to the partial-tracing of any number of modes.

 \begin{figure}
\centering
\includegraphics[width=55 mm]{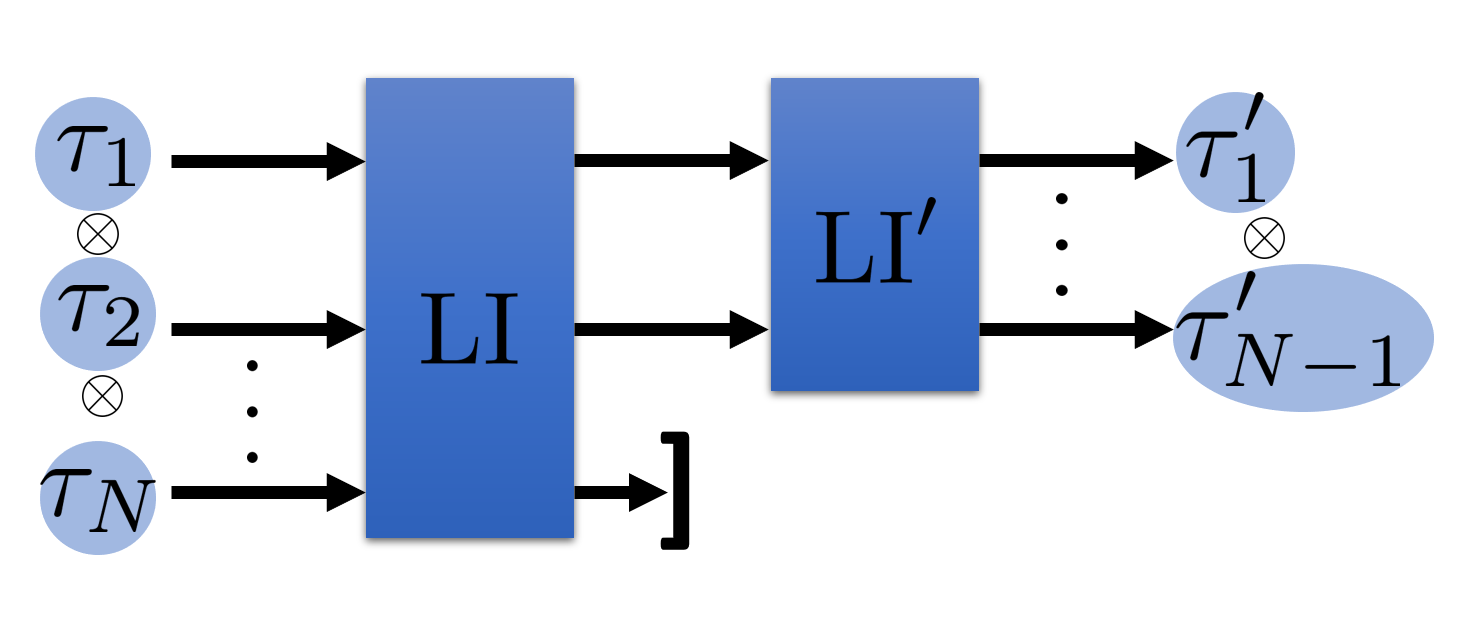}
\caption{Partial-tracing the last mode of an $N$-mode free state $\rho^f_{1\cdots N}$ results into another free state, which can be decomposed back into another product of thermal states $\tau_1'\otimes \cdots \tau_N'$ via another linear interferometer LI'. This results from Cauchy's interlacing theorem and from the structure of  free covariance matrices, Eq. \eqref{eq:free-cov}.}
\label{fig:tensoring-free-state}
\end{figure}

In the following, we denote the set of free operations as $\fo$, consisting of all operations from $O_1$ to $O_3$. It is of interest to note that $\fo$ includes the (free) quantum channels $\Phi_f$, which are generated as
\begin{align}
\label{eq:free-channel}
\Phi_f(\rho^S) = \mathrm{Tr}_{A}\left[ U_{SA}^{PG} \left(\rho^S\otimes\rho_{f}^{A}\right) U_{SA}^{PG\dagger}\right],
\end{align}
where $\rho^{S}$ is the state of the system, $\rho_{f}^{A}$ is a free state of the ancilla, and $U_{SA}^{PG}$ is a passive Gaussian unitary. In particular, if the system and ancilla are single-mode each, one recovers a thermal bosonic channel where the system mode is coupled with an ancilla mode in the thermal states $\tau_{th}^{A}$  (note that the temperature of the thermal state is unspecified here, in contrast to Refs.  \cite{Michal2013, Brandao2013, Brandao2015b}). Let us denote the Kraus operators of this single-mode channel $\Phi_f$ by $\{K_i\}$. If we want to include postselection in the context of our free operations, \textit{i.e.}, if we assume the access to individual Kraus operators, this necessarily demands that $K_i\tau_{th}K_i^\dagger$ is proportional to some thermal state for each index $i$. This extra condition on the Kraus operators does not follow from $\Phi_f$ being a free channel and, in fact, many desirable free channels do not satisfy this extra condition. For instance, consider a single-mode pure-loss channel, which is defined as
\begin{align}
\label{eq:pure-loss-channel}
\Phi_{PL}(\rho^S) = \mathrm{Tr}_{A}\left[ U_{SA}^{BS} \left(\rho^S\otimes\ket{0}\bra{0}^{A}\right) U_{SA}^{BS\dagger}\right],
\end{align}
where $U_{SA}^{BS}$ is a beam-splitter unitary and $\ket{0}$ is the vacuum state. The Kraus operators for such a channel are listed in Ref. \cite{Ivan2011}. It is easy to see that the individual Kraus operators of a pure-loss channel do not map thermal states onto thermal states (see Appendix \ref{append:postselection} for more details on postselection). Therefore, requiring postselection to be free is a very stringent condition on the allowed set of quantum operations which we choose not to consider here.

\section{Local activity monotones}
\label{sec:monotones}

In the resource theory that we define based on $\fs$ and $\fo$ (in the state space picture), {\it local activity} is deemed as a resource as it cannot be created from free states using free operations. As in other resource theories, the local activity  can be quantified using any contractive distance,
\begin{align*}
\act\left(\rho\right) = \min_{\sigma\in\fs} D(\rho,\sigma),
\end{align*}
where $D$ is a contractive distance; \textit{i.e.}, it verifies $D(\Lambda(\rho),\Lambda(\sigma))\leq D(\rho,\sigma)$ with $\Lambda$ being a completely positive trace preserving map. Here, we choose to consider a specific monotone that is based on the relative entropy. If we take $D$ to be the relative entropy, then the {\it relative entropy of local activity} can be defined as
\begin{align*}
\act\left(\rho\right) = \min_{\sigma\in\fs} \infdiv*{\rho}{\sigma}.
\end{align*}
Here the relative entropy $\infdiv*{\rho}{\sigma}=\mathrm{Tr}[\rho (\ln\rho-\ln\sigma]$ if $\mathrm{supp}(\rho)\subseteq \mathrm{supp}(\sigma)$ and $\infty$ otherwise. We now list some of the properties of $\act\left(\rho\right)$.

\bigskip
\noindent
{\it$(P_1)$ Monotonicity of the relative entropy of local activity under free operations:} To see this, let the minimum of the distance for a state $\rho$ be achieved at some free state $\sigma^*\in\fs$. Then,
\begin{align*}
\act\left(\rho\right) &= \infdiv*{\rho}{\sigma^*} \\
&\geq \infdiv*{\fo[\rho]}{\fo[\sigma^*]} \\
&= \infdiv*{\fo[\rho]}{\tilde{\sigma}} \\
&\geq \min_{\sigma\in\fs} \infdiv*{\fo[\rho]}{\sigma} =\act\left(\fo[\rho]\right),
\end{align*}
where $\tilde{\sigma} = \fo[\sigma^*]$ is necessarily some free state. The above is true for any contractive distance $D$.

\bigskip
\noindent
{\it$(P_2)$ Invariance of the relative entropy of local activity under passive Gaussian unitaries:}
Let $U^{\mathrm{PG}}$ be a passive Gaussian unitary transformation; then 
\begin{align}
\act\left( U^{\mathrm{PG}} \rho  U^{\mathrm{PG}\dagger} \right) = \act\left( \rho \right).
\end{align}
It follows from the fact that $\infdiv*{ U^{\mathrm{PG}} \rho U^{\mathrm{PG}\dagger} }{\sigma^*} = \infdiv*{\rho}{U^{\mathrm{PG}\dagger}  \sigma^* U^{\mathrm{PG}}} $ and $U^{\mathrm{PG}\dagger}  \sigma^* U^{\mathrm{PG}}$ is a free state if $\sigma^*$ is a free state. Minimizing over  $U^{\mathrm{PG}\dagger}  \sigma^* U^{\mathrm{PG}}$ is just equivalent to minimizing over free states (the set $\fs$ is closed under passive Gaussian unitaries). \\ \ \\

\noindent
{\it$(P_3)$ Relaxed subadditivity of the relative entropy of local activity:} We are looking for a relation between the relative entropy of local activity of a composite state and the relative entropy of local activity of its marginals. Consider a quantum system composed of $M+N$ modes in  a state $\rho_{AB}$; here $A(B)$ are subsystems of the composite system $AB$ with $M(N)$ modes. Let $\act\left(\rho_A\right) =\infdiv*{\rho_A}{\sigma_A^*}$ and $\act\left(\rho_B\right) =\infdiv*{\rho_B}{\sigma_B^*}$, where $\sigma_A^*$ and $\sigma_B^*$ are the free states that achieve the minima for the respective reduced states $\rho_A$ and $\rho_B$. Then,
\begin{align*}
&\act\left(\rho_A\right)  + \act\left(\rho_B\right)\\
 &=\infdiv*{\rho_A}{\sigma_A^*}+\infdiv*{\rho_B}{\sigma_B^*}\\
&=-\infdiv*{\rho_{AB}}{\rho_A\otimes \rho_B} + \infdiv*{\rho_{AB}}{\sigma_A^*\otimes\sigma_B^*}\\
&\geq -\infdiv*{\rho_{AB}}{\rho_A\otimes \rho_B} + \min_{\sigma_{AB}\in \fs}\infdiv*{\rho_{AB}}{\sigma_{AB}}\\
&= -\infdiv*{\rho_{AB}}{\rho_A\otimes \rho_B} + \act\left(\rho_{AB}\right).
\end{align*}
Therefore,
\begin{align*}
\act\left(\rho_{AB}\right) \leq \act\left(\rho_A\right)  + \act\left(\rho_B\right) +\infdiv*{\rho_{AB}}{\rho_A\otimes \rho_B},
\end{align*}
where $\infdiv*{\rho_{AB}}{\rho_A\otimes \rho_B}=S(\rho_A)+S(\rho_B)-S(\rho_{AB})\ge 0$ is the quantum mutual information between the states $\rho_A$ and $\rho_B$. As a consequence, in the special case where  
$\rho_{AB}=\rho_A\otimes \rho_B$, this ``relaxed'' form of subadditivity translates into subadditivity for product states (see also $P_4$).

The subadditivity of $A_l$ may look undesirable since it translates the fact that a composite system holds fewer resources than the sum of its components. The reason why this is possible is linked to the fact that $\fs$ is not convex. In contrast, as we will see in Sec. \ref{sec:extra-work}, the extractable work under local Gaussian unitaries $W_l$ is superadditive as a consequence of the convexity of $\IW$.

\bigskip
\noindent
{\it$(P_4)$ Additivity of the relative entropy of local activity for the product of single-mode states with zero-displacement vectors:}
We have
\begin{align}
\act\left(\otimes_{i=1}^{m}\rho_i\right) = \sum_{i=1}^{m} \act\left(\rho_i\right). 
\end{align}
The inequality ``$\le$'' is immediate from property $P_3$, but the inequality ``$\ge$'' holds too. The proof of this property is provided in the next section as we need the explicit expression of $\act$ for one-mode states.

\section{Relative entropy of local activity}
\label{sec:mono-calc}

After having defined and established properties of the relative entropy of local activity, let us explicitly calculate its value for some single-mode and two-mode cases.

\subsection{Single-mode case}
For a single mode, the only free states are thermal states at different temperatures. The expression of a thermal state in the Fock basis is given by
\begin{align*}
\tau_{\bar{n}} = \sum_{n=0}^\infty \frac{\bar{n}^n}{(\bar{n}+1)^{n+1}}\ket{n}\bra{n},
\end{align*}
where $\bar{n}$ is the mean photon number.
In the single-mode case, the relative entropy of local activity for a state $\rho$ is given by
\begin{align*}
\act\left(\rho\right) &= \min_{\bar{n}}\infdiv*{\rho}{\tau_{\bar{n}} }\nonumber\\
&= \min_{\bar{n}} \left[-S(\rho)-\mathrm{Tr}\left( \rho \ln \tau_{\bar{n}}   \right)  \right].
\end{align*}
It is easy to check that the minimum appearing in this equation is attained for $\bar{n} = \sum_n n\rho_{nn} = \bar{n}_\rho= \mathrm{Tr}[\rho\hat{a}^\dagger \hat{a}]$, where the $\rho_{nn}=\langle n | \rho | n\rangle$ represent the diagonal elements of $\rho$ in the Fock basis. Thus,
\begin{align*}
\act\left(\rho\right) &= -S(\rho) - \sum_{n=0}^{\infty}\rho_{nn}\left(n\ln \bar{n}_\rho -(n+1)\ln (\bar{n}_\rho+1)\right)\\
&= -S(\rho) +g\left(\bar{n}_\rho+1/2 \right),
\end{align*}
where
\begin{align*}
g\left(\bar{n}_\rho+1/2 \right)= (\bar{n}_\rho + 1) \ln (\bar{n}_\rho + 1)  - \bar{n}_\rho  \ln \bar{n}_\rho
\end{align*}
denotes the entropy of a thermal state $\tau_{\bar{n}_\rho}$ having the same energy as $\rho$. The local activity of a single-mode state is just the relative entropy between the state $\rho$ and the thermal state $\tau_{\bar{n}_\rho}$ with the same energy, namely $\act\left(\rho\right) = \infdiv*{\rho}{\tau_{\bar{n}_\rho}}$. Hence, the activity of a quantum system in a state $\rho$ measures a sort of distance from the thermal state having the same mean photon number. As it happens, the same quantity appears as the definition of the coherence measure of Gaussian states for a single mode \cite{Xu2016} (see also Sec. \ref{sec:comparison-coherence}).

\subsection{Application to the proof of additivity}
We are now in position to prove the additivity property $P_4$ of the relative entropy of local activity for products of single-mode states. We have
\begin{align}
&\act\left(\otimes_{i=1}^{m}\rho_i\right)\nonumber\\
&= \min_{U^{\mathrm{PG}},\bar{n}_1,\cdots,\bar{n}_m}\infdiv*{\bigotimes_{i=1}^{m}\rho_i}{U^{\mathrm{PG}}\bigotimes_{i=1}^{m}\tau_{\bar{n}_i} \, U^{\mathrm{PG}\dagger}}\nonumber\\
&= \min_{U^{\mathrm{PG}},\bar{n}_1,\cdots,\bar{n}_m}\infdiv*{ U^{\mathrm{PG}\dagger}\bigotimes_{i=1}^{m}\rho_i \, U^{\mathrm{PG}}   }{\bigotimes_{i=1}^{m}\tau_{\bar{n}_i}}\nonumber\\
&= \min_{U^{\mathrm{PG}},\bar{n}_1,\cdots,\bar{n}_m} \left[-\sum_{i=1}^{m}S(\rho_i) - \sum_{i=1}^{m}\mathrm{Tr} \left( \tilde{\rho}_i \ln \tau_{\bar{n}_i}  \right)  \right]\nonumber\\
&=  -\sum_{i=1}^{m}S(\rho_i) +\min_{U^{\mathrm{PG}}} \, \sum_{i=1}^{m} \,  g\left(\bar{n}_{\tilde{\rho}_i}+1/2 \right),
\end{align}
where $\tilde{\rho}_{1\cdots m}=U^{\mathrm{PG}\dagger}\bigotimes_{i=1}^{m}\rho_i \, U^{\mathrm{PG}}$ and $\tilde{\rho}_1=\mathrm{Tr}_{\lnot 1}(\tilde{\rho}_{1\dots m})$, and so on. Here, we have used the result of the minimization over $\bar{n}_1,\cdots,\bar{n}_m$ coming from the single-mode case. Now, we show that $\min_{U^{\mathrm{PG}}} \sum_{i=1}^{m}g\left(\bar{n}_{\tilde{\rho}_i}+\frac{1}{2} \right) = \sum_{i=1}^{m}g\left(\bar{n}_{\rho_i}+\frac{1}{2} \right)$. To prove this, let us note that $U^{\mathrm{PG}}$ transforms the annihilation operators as follows:
\begin{align}
\hat{\tilde{a}}_i = \sum_{j=1}^{m} u_{ij} \hat{a}_j,
\end{align}
where $u_{ij}$ are the matrix elements of an arbitrary unitary matrix $u$. Note that $\bar{n}_{\tilde{\rho}_i}=\mathrm{Tr}[\hat{\tilde{a}}_i^\dagger \hat{\tilde{a}}_i \bigotimes_{i=1}^{m}\rho_i] = \langle\hat{\tilde{a}}_i^\dagger \hat{\tilde{a}}_i \rangle$ and $\bar{n}_{\rho_i}=\mathrm{Tr}[\hat{a}_i^\dagger \hat{a}_i \bigotimes_{i=1}^{m}\rho_i] = \langle\hat{a}_i^\dagger \hat{a}_i\rangle$. We have
\begin{align*}
\bar{n}_{\tilde{\rho}_i} +\frac{1}{2} &=  \sum_{j,l=1}^{m} u_{ij}^*   u_{il} \left\langle  \hat{a}_j^\dagger \hat{a}_l \right\rangle+\frac{1}{2}\nonumber\\
&=\sum_{j=1}^{m} u_{ij}^*   u_{ij} \left\langle  \hat{a}_j^\dagger \hat{a}_j \right\rangle + \sum_{\substack{j,l=1\\j\neq l}}^{m} u_{ij}^*   u_{il} \left\langle  \hat{a}_j^\dagger \hat{a}_l \right\rangle+\frac{1}{2}\nonumber\\
&=\sum_{j=1}^{m} u_{ij}^*   u_{ij} \bar{n}_{\rho_j} + \sum_{\substack{j,l=1\\j\neq l}}^{m} u_{ij}^*   u_{il} \left\langle  \hat{a}_j^\dagger\right\rangle \left\langle \hat{a}_l \right\rangle+\frac{1}{2} \nonumber\\
&=\sum_{j=1}^{m} u_{ij}^* u_{ij}  \left(\bar{n}_{\rho_j} +\frac{1}{2}\right)   ,
\end{align*}
where for the last step we have assumed that all the Gaussian states $\rho_i$ $(i=1,\cdots, m)$ have a zero-displacement vector. Now, from the concavity of function $g$, we have 
\begin{align*}
\sum_{i=1}^{m} g\left(\bar{n}_{\tilde{\rho}_i}+\frac{1}{2}\right) &\geq \sum_{i=1}^{m} \sum_{j=1}^{m} u_{ij}^*   u_{ij} \,  g\left(\bar{n}_{\rho_j}+\frac{1}{2}\right)\nonumber\\
&= \sum_{j=1}^{m}   \left(\sum_{i=1}^{m} u_{ij}^*   u_{ij} \right)g\left(\bar{n}_{\rho_j}+\frac{1}{2}\right)\nonumber\\
&= \sum_{j=1}^{m} g\left(\bar{n}_{\rho_j}+\frac{1}{2}\right),
\end{align*}
where equality is achieved when $U^{PG}=\mathbb{I}$ or $u_{ij}=\delta_{ij}$.
Thus, applying a passive Gaussian unitary $U^{\mathrm{PG}}$ can only increase the value of $\sum_{i=1}^{m} g\left(\bar{n}_{\tilde{\rho}_i}+\frac{1}{2}\right)$. 
Hence, using $P_3$ and above, we have
\begin{align*}
\act\left(\otimes_{i=1}^{m}\rho_i\right) &=  -\sum_{i=1}^{m}S(\rho_i) + \sum_{i=1}^{m}g\left(\bar{n}_{\rho_i}+1/2\right)\\
&=\sum_{i=1}^{m} \act\left(\rho_i\right).
\end{align*}

\subsection{Two-mode case}
Consider a Gaussian state $\rho_1({\bf d}, \pmb\Gamma_1)$ of two modes with covariance matrix
\begin{align*}
\pmb\Gamma_1=\begin{pmatrix}
A & C\\
C^T& B
\end{pmatrix},
\end{align*}
and the displacement vector ${\bf d}=\begin{pmatrix}
d_1&d_2&d_3&d_4
\end{pmatrix}^T$.
Define $\alpha=\mathrm{Tr}[A]$, $\beta=\mathrm{Tr}[B]$, $c=\mathrm{Tr}[C]$, and $\upsilon=-\mathrm{Tr}[\omega C]$.  Furthermore, define $\tilde{\alpha}=(\alpha+\beta+\tilde{d}_1)$, $\tilde{\beta}=(\alpha-\beta+\tilde{d}_2)$, $\tilde{c}=(c+\tilde{d}_3)$ and $\tilde{\upsilon}=(\upsilon+\tilde{d}_4)$, where $\tilde{d}_1=(d_1^2+d_2^2+d_3^2+d_4^2)$, $\tilde{d}_2=(d_1^2+d_2^2-d_3^2-d_4^2)$, $\tilde{d}_3=(d_1d_3+d_2d_4)$, and $\tilde{d}_4=(d_1d_4-d_2d_3)$.  Finally, let us define 
\begin{align*}
\mathcal{S}&=\begin{pmatrix}
R_1 & 0\\
0& R_2 
\end{pmatrix} \begin{pmatrix}
\cos\theta \mathbb{I}_2 & \sin\theta\mathbb{I}_2\\
-\sin\theta \mathbb{I}_2 &\cos\theta \mathbb{I}_2 
\end{pmatrix} \begin{pmatrix}
R_3 & 0\\
0& R_4 
\end{pmatrix} \\
&=\begin{pmatrix}
\cos\theta R_1 R_3 & \sin\theta R_1 R_4\\
-\sin\theta R_2 R_3 &\cos\theta R_2 R_4 
\end{pmatrix},
\end{align*}
and for $i=1,\cdots,4$
\begin{align*}
R_i=\begin{pmatrix}
\cos\phi_i& \sin\phi_i\\
-\sin\phi_i&\cos\phi_i 
\end{pmatrix},
\end{align*}
with $\delta\phi=\phi_1-\phi_2$. We see that the angles $\phi_3$ and $\phi_4$ do not actually matter in the definition of the most general case. We show in Appendix \ref{append:mono-two} that the closest free state to $\rho_1({\bf d},\pmb\Gamma_1)$ corresponds to 
\begin{align*}
\pmb\Gamma_{\mathrm{free}}=\mathcal{S}\begin{pmatrix}
b_1\mathbb{I}_2 & 0\\
0& b_2 \mathbb{I}_2
\end{pmatrix} \mathcal{S}^T,
\end{align*}
with the following definitions:
\begin{align*}
&b_1= \frac{1}{4} \left[\tilde{\alpha}+\sqrt{\tilde{\beta}^2+4(\tilde{c}^2+\tilde{\upsilon}^2)}\right],\\
&b_2= \frac{1}{4} \left[\tilde{\alpha}-\sqrt{\tilde{\beta}^2+4(\tilde{c}^2+\tilde{\upsilon}^2)}\right],\\
&\theta=-\frac{1}{2}\arctan\left( \frac{2\sqrt{\tilde{c}^2+\tilde{\upsilon}^2}}{\tilde{\beta}} \right),\\
&\delta\phi=\arctan(\tilde{\upsilon}/\tilde{c}).
\end{align*} 
In this case, the relative entropy of local activity is given by
\begin{align*}
\act\left(\rho\right)&=\sum_{i=1}^2 \left[g(b_i)-g(\nu_i)\right],
\end{align*} 
where the $\nu_i$ are the symplectic eigenvalues of $\rho$. The details of the calculations are provided in Appendix \ref{append:mono-two}.
%
%

\subsection{Examples}

In the following, we give some examples of values of the relative entropy of local activity for some quantum states of interest.

\bigskip
\noindent
{\it$(1)$ Fock states:} 
The Fock state $\ket{n}$ has a value of the relative entropy of local activity given by $\act= g(n+\frac{1}{2})$.

\bigskip
\noindent
{\it$(2)$ Squeezed state:} 
The squeezed vacuum state $\ket{\mathrm{Sq}}$ of squeezing parameter $r$ has a value of the relative entropy of local activity given by $\act= g(\sinh^2 r+\frac{1}{2})$.

\bigskip
\noindent
{\it$(3)$ Coherent state:} 
The coherent state $\ket{\alpha}$ of complex amplitude $\alpha$ has a value of the relative entropy of local activity given by $\act= g(|\alpha|^2+\frac{1}{2})$.

\bigskip
\noindent
{\it$(4)$ Two-mode squeezed vacuum state:} 
The two-mode squeezed vacuum state $\ket{\mathrm{TMS}}$ of squeezing parameter $r$ has a value of the relative entropy of local activity given by $\act= 2 \, g(\sinh^2 r + \frac{1}{2})$. This results from the fact that it can be obtained with a 50:50 beam splitter (\textit{i.e.}, a passive Gaussian unitary) applied on a product of two single-mode squeezed vacuum states (with squeezing of orthogonal quadratures).

\subsection{Comparison with coherence of Gaussian states}
\label{sec:comparison-coherence}
The relative entropy of coherence for Gaussian states was defined in Ref. \cite{Xu2016}. For the sake of clarity, we remind the reader of its expression here. For an $N$-mode Gaussian state $\rho$, the Gaussian (relative entropy of) coherence is given by \cite{Xu2016}
\begin{align*}
\mathcal{C}_G(\rho)=\sum_{i=1}^N \left[g\left(n_{\rho_i} + \frac{1}{2}\right)-g(\nu_i)\right],
\end{align*} 
where the $\nu_i$ are the symplectic eigenvalues of $\rho$ and $n_{\rho_i} =\mathrm{Tr}[\hat{a}_i^\dagger \hat{a}_i \rho]$ is the mean number of photons of the state $\rho_i$. We note that
\begin{align*}
\act\left(\rho\right) = \min_{U^{\mathrm{PG}},\tau_1,\cdots,\tau_N} \infdiv*{\rho}{U^{\mathrm{PG}}(\tau_1\otimes\cdots\otimes \tau_N)U^{\mathrm{PG}\dagger}}\\
\leq \min_{\tau_1,\cdots,\tau_N} \infdiv*{\rho}{\tau_1\otimes\cdots\otimes \tau_N}
:=\mathcal{C}_G(\rho).
\end{align*}
Also, we have shown that $\mathcal{C}_G(\rho)$ coincides with the relative entropy of local activity in the single-mode case. This is precisely because our free states and the Gaussian incoherent states coincide in the case of one mode. However, for more than one mode, the two monotones are different in general. Still, in the special case of an $N$-mode state $\rho(N)$ that can be written as $U^\mathrm{PG}(\rho_1\otimes\cdots\otimes \rho_N)U^{\mathrm{PG}\dagger}$, we recover $\act\left(\rho(N)\right)=\mathcal{C}_G(\rho(N))$. For instance, this happens to be the case for a two-mode squeezed vacuum state.

\section{Extractable work under local Gaussian unitaries}
\label{sec:extra-work}
It is known that passive states are those states from which no work can be extracted unitarily \cite{Pusz1978, Lenard78}. Similarly, if we restrict ourselves to extracting work using only Gaussian unitary transformations, the Gaussian-passive states are the states from which no work can be extracted using Gaussian unitaries (the Gaussian-passive states and their characterization have been presented in Ref. \cite{Brown2016}). Here, we consider local Gaussian work extraction scenarios as pictured in Fig. \ref{fig:scenario}, where multiple parties try to extract work from a multimode system via local Gaussian unitary transformations (i.e., single-mode squeezers) assisted with a global linear interferometer (LI). We will define a measure of the maximum local Gaussian extractable work, noted $W_l$,  and show that it is a monotone under free operations $\fo$. From now on, we rather use the phase-space picture and work with displacement vectors and covariance matrices, which is natural since the energy of a state only depends on first- and second-order moments of the field operators. Let us consider an $N$-mode arbitrary state $\rho (\pmb\Gamma,\bf {\bar{x}})$ with covariance matrix $\pmb\Gamma$ and displacement vector $\bf {\bar{x}}$. The local Gaussian extractable work assisted with LIs is defined as
\begin{align*}
&{W}_{l}(\rho)\\
&= \max_{U^{PG},\mathcal{D}_\alpha,U_{sq}}\mathrm{Tr}\left[H\left(\rho-U_{sq} U^{PG} \mathcal{D}_\alpha  \rho \mathcal{D}^{\dagger}_\alpha U^{PG\dagger} U_{sq}^{\dagger} \right)\right]\\
&= \max_{U^{PG},\mathcal{D}_\beta,U_{sq}}\mathrm{Tr}\left[H\left(\rho- \mathcal{D}_\beta U_{sq}  U^{PG} \rho U^{PG\dagger} U_{sq}^{\dagger}\mathcal{D}^{\dagger}_\beta \right)\right],
\end{align*} 
where the maximum is taken over all LIs $(U^{PG})$, single-mode squeezers $(U_{sq})$ and displacement operators ($\mathcal{D}_\alpha$ or $\mathcal{D}_\beta$). The second line follows from the fact that if we exchange a displacement operator with parameter $\alpha$ with a combination of squeezers and/or LIs, it remains a displacement operator (with another parameter $\beta$). This means that the application of displacement operators before or after the squeezers does not change the maximum local Gaussian extractable work. Thus, in order to extract work locally using Gaussian unitaries, we may first apply with no loss of generality a displacement operator on each mode and extract the available work due to $\bf {\bar{x}}$, thereby making $\bf {\bar{x}}={\bf 0}$ for subsequent work extraction. Hence, the component of the extractable work due to displacements is trivial and may be disregarded for simplicity.



The energy of an $N$-mode state $\rho (\pmb\Gamma,\bf {\bar{x}})$ is given by Eq. \eqref{eq:ene-exp}, and, in particular, the energy of $\rho (\pmb\Gamma,\bf {0})$ is given by $\mathrm{Tr}[\pmb\Gamma]/2$.  We will now deduce from it an expression for the local extractable work $\mathcal{W}_{l}(\pmb\Gamma)$ in  phase space (we use the notation $\mathcal{W}_{l}$ when it is written as a function of the CM and $W_l$ when it is a function of the state). A particularly relevant tool in this analysis will be the so-called Bloch-Messiah decomposition \cite{Weedbrook2012, Arvind1995, Braunstein2005} of symplectic matrices, which states that for any symplectic matrix ${\bf\mathcal{S}}$, ${\bf\mathcal{S}}= {\bf O}_1 \oplus_{i} \! {\bf\mathfrak{S}}(r_i) \, {\bf O}_2$, where ${\bf O}_1$ and ${\bf O}_2$ are orthogonal symplectic matrices (\textit{i.e.}, LIs) and ${\bf\mathfrak{S}}(r_i)$ are single mode squeezers, namely
\begin{align*}
{\bf \mathfrak{S}}(r_i)=\begin{pmatrix}
e^{r_i} & 0\\
0 & e^{-r_i}
\end{pmatrix}.
\end{align*}
Another useful property is that, given any $2N\times 2N$ positive definite-matrix $\pmb\Gamma$ \cite{Bhatia2015}
\begin{align}
\label{eq:bhatia}
\min_{\bf{\mathcal{S}}:\bf \mathcal{S}\pmb\Omega\mathcal{S}^T=\pmb\Omega} \mathrm{Tr}[\mathcal{S}\pmb\Gamma\mathcal{S}^T]=\mathrm{Str}[\pmb\Gamma] \equiv 2\sum_{i=1}^{N}\nu_i,
\end{align}
where $\mathrm{Str}[\pmb\Gamma]$ represents the symplectic trace of $\pmb\Gamma$, and is equal to twice the sum of the symplectic eigenvalues $\nu_i$.
Using this, the local Gaussian extractable work assisted with global LIs from an $N$-mode Gaussian state with covariance matrix $\pmb\Gamma$ (and displacement vector $\bf {\bar{x}}={\bf 0}$) is defined as
\begin{align}
\label{eq:work-def}
&\mathcal{W}_{l}(\pmb\Gamma)\nonumber\\
&:=\frac{1}{2}\max_{{\bf O},\{{\bf\mathfrak{S}}(r_i)\}}\mathrm{Tr}\left[ \pmb\Gamma - \oplus_{i}\mathfrak{S}(r_i) \, {\bf O}\pmb\Gamma {\bf O}^T\oplus_{i}\!\mathfrak{S}(r_i) \right]\nonumber\\
&=\frac{1}{2}\left(\mathrm{Tr}\left[ \pmb\Gamma\right] - \mathrm{Str}\left[\pmb\Gamma \right]\right)\nonumber\\
&=\frac{1}{2}\mathrm{Tr}\left[ \pmb\Gamma - \pmb\Gamma_{th} \right]=\frac{1}{2}\sum_{i=1}^{2N}\lambda_i- \sum_{i=1}^{N}\nu_i,
\end{align}
where $\lambda_i$ and $\nu_i$ are the eigenvalues and symplectic eigenvalues of $\pmb\Gamma$, respectively. The second line in the chain of equalities above follows from the Bloch-Messiah decomposition and Eq. \eqref{eq:bhatia}. In particular,
\begin{align}
&\min_{{\bf O},\{{\bf\mathfrak{S}}(r_i)\}}\mathrm{Tr}\left[\oplus_{i}\mathfrak{S}(r_i) \, {\bf O}\pmb\Gamma {\bf O}^T\oplus_{i}\!\mathfrak{S}(r_i) \right]\nonumber\\
& = \min_{\bf{\mathcal{S}}:\bf \mathcal{S}\pmb\Omega\mathcal{S}^T=\pmb\Omega} \mathrm{Tr}[\mathcal{S}\pmb\Gamma\mathcal{S}^T]=\mathrm{Str}[\pmb\Gamma].
\end{align}
The third line in Eq. \eqref{eq:work-def} follows from the fact that the symplectic trace of the covariance matrix $\pmb\Gamma$ is equal to the trace of the covariance matrix $\pmb\Gamma_{th}$ of the product of thermal states having the same symplectic spectrum.

Next, we list some of the important properties of the local extractable work $\mathcal{W}_{l}(\pmb\Gamma)$.
\bigskip
\newline
\noindent
{\it ($P'_1$) The local extractable work $\mathcal{W}_{l}(\pmb\Gamma)$ is positive semidefinite:} Since we have
$$\mathrm{Str}[\pmb\Gamma]=\min_{\bf{\mathcal{S}}:\bf \mathcal{S}\pmb\Omega\mathcal{S}^T=\pmb\Omega} \mathrm{Tr}[\mathcal{S}\pmb\Gamma\mathcal{S}^T] \leq \mathrm{Tr}[\pmb\Gamma],$$ 
it is clear that $\mathcal{W}_{l}(\pmb\Gamma)\geq 0$.
\bigskip
\newline

\noindent
{\it ($P'_2$) The local extractable work $\mathcal{W}_{l}(\pmb\Gamma)$ is a convex function of covariance matrices $\pmb\Gamma$:}
From Eq. \eqref{eq:bhatia}, it follows that given two positive-definite matrices $C$ and $D$ \cite{Bhatia2015}
\begin{align}
\label{eq:conc-str}
\mathrm{Str}[C+D]\geq \mathrm{Str}[C]+\mathrm{Str}[D].
\end{align}
This property of the symplectic trace implies immediately that $\mathcal{W}_{l}(\pmb\Gamma)$ is convex; \textit{i.e.}, 
if a covariance matrix is given by $\pmb\Gamma=\sum_{j=1}^m p_j \pmb\Gamma^{(j)}$, with $p_j\ge 0$ and $\sum_j p_j=1$, then 
\begin{align}
&\mathcal{W}_{l}\left(\sum_{j=1}^m p_j \pmb\Gamma^{(j)}\right)\nonumber\\
&=\frac{1}{2}\left( \sum_{j=1}^m p_j \mathrm{Tr}\left[\pmb\Gamma^{(j)}\right] - \mathrm{Str}\left[\sum_{j=1}^m p_j \pmb\Gamma^{(j)} \right]\right)\nonumber\\
&\leq \sum_{j=1}^m p_j \, \mathcal{W}_{l}\left(\pmb\Gamma^{(j)}\right).
\end{align}
This simply means that if we ``forget" which term $j$ the state belongs to within a convex mixture, this can only reduce the maximum local extractable work.

Note also that for a free state as defined in Sec. \ref{sec:basic}, \textit{i.e.}, a Gaussian state with covariance matrix of the form $\pmb\Gamma_{free}={\bf O} \pmb\Gamma_{th}{\bf O}$, the symplectic eigenvalues and eigenvalues of the covariance matrix coincide, so that $\mathcal{W}_{l}\left(\pmb\Gamma_{free}\right)=0$. Moreover, from the convexity of $\mathcal{W}_{l}$, we have
\begin{align*}
\mathcal{W}_{l}\left(\sum_{j=1}^m p_j \pmb\Gamma^{(j)}_{free}\right)=0,
\end{align*}
so the set of free covariances matrices is convex.
We note here that $\sum_{j=1}^m p_j \pmb\Gamma^{(j)}_{free}$ can always be written as $\pmb\Gamma_{free}={\bf O}\pmb\Gamma_{th}{\bf O}^T$, where ${\bf O}$ is some orthogonal symplectic matrix.
\bigskip
\newline

\noindent
{\it ($P'_3$) The local extractable work $\mathcal{W}_{l}\left(\pmb\Gamma\right)=0$ if and only if $\pmb\Gamma= {\bf O}_2\pmb\Gamma_{th}{\bf O}_2^T$:} As already mentioned in $P'_2$, if $\pmb\Gamma= {\bf O}_2\pmb\Gamma_{th}{\bf O}_2^T$,  then trivially $\mathcal{W}_{l}\left(\pmb\Gamma\right)=0$. For proving the converse, we use the Williamson theorem and  Bloch-Messiah decomposition in order to write $\pmb\Gamma={\bf O}_1 \bigoplus_{i}{\bf\mathfrak{S}}(r_i){\bf O}_2\pmb\Gamma_{th}{\bf O}_2^T \bigoplus_{i}{\bf\mathfrak{S}}(r_i){\bf O}_1^T$, so we have
\begin{align*}
&\mathcal{W}_{l}\left(\pmb\Gamma\right)=0\\
&\Rightarrow \mathrm{Tr}\left[ \oplus_{i}{\bf\mathfrak{S}}(r_i) \, {\bf O}_2\pmb\Gamma_{th}{\bf O}_2^T \oplus_{i}\!{\bf\mathfrak{S}}(r_i) \right] - \mathrm{Tr}[\pmb\Gamma_{th}]=0\\
&\Rightarrow \oplus_{i}{\bf\mathfrak{S}}(r_i)=\mathbb{I}_{2N}.
\end{align*}
The last implication follows from the observation that the energy of a state of the form ${\bf O}_2\pmb\Gamma_{th}{\bf O}_2^T$ is always increased if we apply some single-mode squeezers on it. Therefore, the local extractable work vanishes if and only if the covariance matrix is free, namely
\begin{align}
&\mathcal{W}_{l}\left(\pmb\Gamma\right)=0\Leftrightarrow \pmb\Gamma= {\bf O}_2\pmb\Gamma_{th}{\bf O}_2^T.
\end{align}

\bigskip
\noindent
{\it ($P'_4$) The local extractable work  $\mathcal{W}_{l}(\pmb\Gamma)$ is superadditive:}  In other words, 
$$\mathcal{W}_{l}(\pmb\Gamma)\geq \mathcal{W}_{l}(\pmb\Gamma_A)+\mathcal{W}_{l}(\pmb\Gamma_B),$$ where the $N$-mode covariance matrix $\pmb\Gamma$ is partitioned into two subsets $A$ and $B$ consisting of $m$ and $(N-m)$ modes, respectively, that is,
\begin{align}
\pmb\Gamma=\begin{pmatrix}
\pmb\Gamma_A & \pmb\Gamma_{AB}\\
\pmb\Gamma_{AB}^{T}& \pmb\Gamma_B .
\end{pmatrix}
\end{align}
Without loss of generality, we can find two symplectic matrices $\mathcal{S}_A$ and $\mathcal{S}_B$ such that 
\begin{align}
(\mathcal{S}_A \oplus\mathcal{S}_B)\pmb\Gamma (\mathcal{S}_A \oplus\mathcal{S}_B)^T=\begin{pmatrix}
\pmb\Gamma_A^{th} & \pmb\Gamma^{'}_{AB}\\
\pmb\Gamma_{AB}^{'T}& \pmb\Gamma_B^{th}
\end{pmatrix}:=\pmb\Gamma'.
\end{align}
 We have
\begin{align}
\label{eq:work-sup-add}
\mathcal{W}_{l}(\pmb\Gamma) &= \mathrm{Tr}\left[ \pmb\Gamma \right] -\mathrm{Str}\left[\pmb\Gamma\right]\nonumber\\
&=\mathrm{Tr}\left[ \pmb\Gamma_A \right]+\mathrm{Tr}\left[ \pmb\Gamma_B \right]-\mathrm{Str}\left[\pmb\Gamma'\right]\nonumber\\
&\geq \mathrm{Tr}\left[ \pmb\Gamma_A \right]+\mathrm{Tr}\left[ \pmb\Gamma_B \right]-\mathrm{Tr}\left[\pmb\Gamma'\right]\nonumber \\
&= \mathrm{Tr}\left[ \pmb\Gamma_A \right] -\mathrm{Tr}\left[\pmb\Gamma_A^{th}\right] +\mathrm{Tr}\left[ \pmb\Gamma_B \right]-\mathrm{Tr}\left[\pmb\Gamma_B^{th}\right]\nonumber\\
&=\mathcal{W}_{l}(\pmb\Gamma_A)+\mathcal{W}_{l}(\pmb\Gamma_B).
\end{align}
In the second line, we have used the invariance of the symplectic trace under symplectic transformations. In the third line, we used the fact that $\mathrm{Str}[\pmb\Gamma']\leq \mathrm{Tr}[\pmb\Gamma']$, which follows from Eq. \eqref{eq:bhatia}. It is easy to see that for product states, i.e., $\pmb\Gamma=\pmb\Gamma_A\oplus\pmb\Gamma_B$, the equality holds. Thus, we have shown that $\mathcal{W}_{l}(\pmb\Gamma)$ is superadditive, which reflects the fact that more work can potentially be extracted from the joint system rather than from its two components separately. Indeed, we can use a LI involving all $n$ modes instead of two separate LIs on the first $m$ modes and last $(N-m)$ modes. Furthermore, the superadditivity property implies that   
\begin{align*}
\mathcal{W}_{l}(\pmb\Gamma) \geq \sum_{i=1}^{N} \mathcal{W}_{l}\left(\pmb\Gamma_i\right),
\end{align*}
where $\pmb\Gamma_i$ is the covariance matrix of the $i$th mode.

\bigskip
\noindent
{\it ($P'_5$) $\mathcal{W}_{l}(\pmb\Gamma)$ is monotonically decreasing under free operations $\fo$:} 
We have seen that $\mathcal{W}_{l}\left( \pmb\Gamma \right)$ is a faithful resource measure, \textit{i.e.}, $\mathcal{W}_{l}\left( \pmb\Gamma \right)\geq 0$, while the equality holds if and only if the CM is free, $\pmb\Gamma= {\bf O}_2\pmb\Gamma_{th}{\bf O}_2^T$.
Now, to obtain a really meaningful resource measure, it remains to be checked that $\mathcal{W}_{l}(\pmb\Gamma)$ decreases under free operations as defined in Sec. \ref{sec:basic}. First, it is obvious to see that $\mathcal{W}_{l}(\pmb\Gamma)=\mathcal{W}_{l}({\bf O}\pmb\Gamma {\bf O}^T)$, where ${\bf O}$ is an orthogonal symplectic matrix. Next, we have seen that 
$\mathcal{W}_{l}(\pmb\Gamma_A\oplus\pmb\Gamma_B)= \mathcal{W}_{l}(\pmb\Gamma_A)+\mathcal{W}_{l}(\pmb\Gamma_B)$, which ensures that tensoring cannot increase the extractable work. Finally, we need to prove that $\mathcal{W}_{l}\left( \pmb\Gamma \right)$ monotonically decreases under partial tracing. This follows trivially from the superadditivity property $P'_4$ (which is actually a stronger result than monotonicity under partial tracing). Thus, we have shown that $\mathcal{W}_{l}(\pmb\Gamma)$ is a superadditive measure that is monotonic under partial tracing. 

\bigskip
To summarize this section, we have defined the extractable work $W_l$ or $\mathcal{W}_{l}$ via local Gaussian unitaries (single-mode squeezers) assisted with global linear interferometers (LIs) from a multimode bosonic system, and have proved that this is a monotone under the free operations $\fo$ of our resource theory of local activity; see Sec. \ref{sec:basic}. Our work extraction scenario (see Fig. \ref{fig:scenario}) can be understood operationally as follows. Given a bosonic system in a quantum state (Gaussian or not) with covariance matrix $\pmb\Gamma$ and displacement vector $\bf\bar{x}$, 
we first apply local displacements to extract work that is there due to $\bf\bar{x}$. This step makes $\bf\bar{x}=0$. Since we are allowed to use LIs before extracting work, we can appropriately convert (using the Bloch-Messiah decomposition) the given covariance matrix $\pmb \Gamma \xrightarrow[\text{}]{\text{LI}} \left(\oplus_{i}{\bf \mathfrak{S}}(r_i) \, {\bf O}_2\pmb\Gamma_{th} {\bf O}_2^T\oplus_{i}\! {\bf \mathfrak{S}}(r_i)\right)$, where $\pmb\Gamma_{th}$ has the same symplectic spectrum as $\pmb\Gamma$. Now, the maximal Gaussian local work is obtained from such a state by applying single-mode squeezers $\bigoplus_{i}{\bf \mathfrak{S}}(r_i)$. After this last step, no further work can be obtained using local Gaussian unitaries, \textit{i.e.}, squeezers, and the covariance matrix becomes free, namely ${\bf O}_2\pmb\Gamma_{th} {\bf O}_2^T$.

\section{Distillation}
\label{sec:distillation}
One of the key advantages offered by the resource-theoretic paradigm of any resource is the succinct description of its interconversion (specifically, its distillation and formation). However, not any interconversion of resources is possible in every resource theory \cite{Brandao2015, Lami2018}. For instance, it is well known that Gaussian entanglement cannot be distilled using Gaussian local operations and classical communication only \cite{Eisert2002, Fiurasek2002, Giedke2002}. More recently, under reasonable assumptions, a general no-go theorem for the distillation of Gaussian resources under resource nongenerating operations was proved \cite{Lami2018}. In this section, we show the possibility of distillation of resources considered in the present work and discuss the consequences it entails, as well as some open problems.
\subsection{For local activity}
Let us first consider the following question. Is it possible to have a deterministic transformation
\begin{align*}
\rho\otimes\rho \xrightarrow[\mathrm{free ~operations}]{?} \sigma,
\end{align*}
where $\rho$ and $\sigma$ are both single-mode states such that $A_l(\sigma)>A_l(\rho)$? We show that this is not possible. Let us consider that the covariance matrix of $\rho$ is given by a $2\times 2$ matrix $\pmb\gamma$, so that the covariance matrix of $\rho\otimes\rho$ is given by
\begin{align*}
\pmb\Gamma (\rho\otimes\rho)=\begin{pmatrix}
\pmb\gamma & 0_{2\times 2} \\
0_{2\times 2} & \pmb\gamma
\end{pmatrix}.
\end{align*}
The free operation one can apply on the two copies of $\rho$ is just a beam splitter with phase shifters, and is represented by 
\begin{align*}
\mathcal{S}=\begin{pmatrix}
\cos\theta R_1 R_3 & \sin\theta R_1 R_4\\
-\sin\theta R_2 R_3 & \cos\theta R_2 R_4
\end{pmatrix},
\end{align*}
where the $R_i$'s are phase shifters for each $i=1,\cdots,4$. Now, let $\pmb\Gamma' =\mathcal{S}\pmb\Gamma(\rho\otimes\rho)\mathcal{S}^T$ and denote by $\pmb\gamma'_1$ and $\pmb\gamma'_2$ the local covariance matrices of the output states $\rho_1'$ and $\rho_2'$ of mode one and mode two, respectively. Then,
 \begin{align}
 \label{eq:loc-cov-st} 
&R_1^T\pmb\gamma'_1 R_1= \left( \cos^2\theta R_3 \pmb\gamma R_3^T + \sin^2\theta R_4 \pmb\gamma R_4^T\right),\nonumber\\
&R_2^T\pmb\gamma'_2R_2= \left(\sin^2\theta R_3 \pmb\gamma R_3^T + \cos^2\theta R_4 \pmb\gamma R_4^T \right).
\end{align}
It is clear that the mean number of photons in state $\rho_1'$ (of CM $\pmb\gamma_1'$) is the same as that of the initial state $\rho$ (of CM $\pmb\gamma$). Therefore,
\begin{align*}
&\act(\rho'_1) - \act(\rho) = S(\rho)-S(\rho'_1)\leq 0.
\end{align*}
The last inequality follows from the linear entropy power inequality for beam splitters. Obviously, a similar inequality is true for the state having covariance matrix $\pmb\gamma'_2$. This shows that it is impossible to increase the local activity of any one-mode state under free operations starting from two copies of the one-mode state. However, we show in the following that this is not the case if we start from two copies of a two-mode states.

We now provide an example which shows that the deterministic distillation of local activity from a two-mode state is in principle possible, starting from two copies of that state. We consider two copies of a nonfree two-mode Gaussian state $\rho$ with relative entropy of local activity $A_l(\rho)$, and show that by using free operations on the four modes, one can generate a two-mode state $\sigma$ with relative entropy of local activity $A_l(\sigma)$ such that $A_l(\sigma)>A_l(\rho)$. In particular, let us start with a two-mode Gaussian state that is in a product of a squeezed thermal state and a vacuum state, \textit{i.e.}, a state with covariance matrix $\pmb\Gamma (\rho)$ defined as
\begin{align*}
\pmb\Gamma (\rho)=\frac{1}{2}\begin{pmatrix}
1 & 0 & 0 & 0\\
0 & 16 &0 & 0\\
0 & 0 & 1 & 0\\
0 & 0 & 0 & 1
\end{pmatrix}.
\end{align*}
Now, the free operation that we choose to apply on $\rho$ is a passive Gaussian unitary $U$ which acts on the annihilation operators of the four modes as
\begin{align*}
U=\frac{1}{2}\begin{pmatrix}
1 & 1 & 1& 1\\
1 & -i &-1 & i\\
1 & -1 & 1 & -1\\
1 & i & -1 & -i
\end{pmatrix}.
\end{align*}
Let $\mathcal{O}$ be the corresponding symplectic transformation, so that $\pmb\Gamma'=\mathcal{O}\pmb\Gamma(\rho\otimes\rho)\mathcal{O}^T$. The first two-mode state has a covariance matrix $\pmb\Gamma (\sigma)$ that is given by
\begin{align*}
\pmb\Gamma (\sigma)=\bigoplus_{i=1}^{2}\frac{1}{4}\begin{pmatrix}
2 & 0 \\
0 & 17
\end{pmatrix}.
\end{align*}
It is easy to check that $A_l(\rho)=0.7621$ and $A_l(\sigma)=1.0019>A_l(\rho)$. This example tells us that it is in principle possible to distill the relative entropy of local activity. This begs us to define a standard unit resource for distillation purposes. However, the (non)existence of a standard unit resource is a mathematically involved problem and requires a separate exposition which we leave for future explorations.


Now, if we assume that it is possible to convert deterministically $n$ copies of a state $\rho$ into $m$ copies of state, say $\sigma$, using free operations, then the rate $\mathcal{R}_{\rho\rightarrow \sigma}$ of conversion can be defined as
\begin{align}
\mathcal{R}_{\rho\rightarrow \sigma}= \frac{m}{n} \leq \frac{\act{(\rho)}}{\act{(\sigma)}}.
\end{align}
The inequality above follows from the additivity and monotonicity of the relative entropy of activity. It will be very interesting to consider the asymptotic scenario for conversion, prove its existence, and show that the optimal asymptotic rate of conversion is given by a ratio of relative entropies of activity. However, we will not deal with these questions here. Moreover, with a probabilistic transformation (if possible), it is possible to increase the local activity of a one-mode state with a free operation starting from two copies of the one-mode state. For example, starting from two copies of a single-photon state $\ket{1}$, we can send them into a $50:50$ beam splitter (free operation) and postselect the second mode onto the vacuum state $\ket{0}$, which is a free state. The outcome is the Fock state $\ket{2}$, which has local activity $g(2+1/2)$ greater than that of the initial states $g(1+1/2)$. However, the existence of free Kraus operators that will allow postselection over free states is an issue we leave open here.

\subsection{For local Gaussian extractable work}
Just like in the case of the relative entropy of activity, the local Gaussian extractable work assisted with linear interferometers cannot be distilled from two copies of a single-mode state. In order to see this, recall that the local covariance matrices after processing the two copies of an initial state through a beam splitter and phase shifters are given by Eq. \eqref{eq:loc-cov-st}. Now, we have 
\begin{align*}
\mathcal{W}_{l}(\pmb\gamma'_1) &= \mathrm{Tr}\left[ \pmb\gamma'_1 \right] -\mathrm{Str}\left[\pmb\gamma'_1\right]\\
&\leq \mathrm{Tr}\left[ \pmb\gamma \right] -\mathrm{Str}\left[\pmb\gamma\right]=\mathcal{W}_{l}(\pmb\gamma).
\end{align*} 
The inequality in the equation above comes from the use of Eq.\eqref{eq:conc-str}. This shows that the two copies of single-mode states are useless for deterministic work extraction. However, we again have a simple example of a two-mode state, from two copies of which we can distill some work. Consider a two-mode state with covariance matrix $\pmb\Gamma=\pmb\gamma_A\oplus\pmb\gamma_B$ such that $\mathcal{W}_{l}(\pmb\gamma_A) > \mathcal{W}_{l}(\pmb\gamma_B)$. The two copies of this state will be denoted by the covariance matrix $\pmb\Gamma^{\oplus 2}=(\pmb\gamma_A\oplus\pmb\gamma_B)\oplus(\pmb\gamma_A\oplus\pmb\gamma_B)$. Now, we can apply a swap between auxiliary modes two and three (using a beam splitter) to get a covariance matrix $\tilde{\pmb\Gamma}^{\oplus 2}=(\pmb\gamma_A\oplus\pmb\gamma_A)\oplus(\pmb\gamma_B\oplus\pmb\gamma_B)$. One can see that the first two modes provide an amount of work that is greater than the work that might be obtained from the initial state $\pmb\Gamma$. Again, the tasks of designing optimal deterministic conversion protocols and probabilistic protocols are left open and will be studied separately.

\section{Conclusion}
\label{sec:conclusion}

In this work, we have explored a possible multipartite extension of a resource theory for quantum thermodynamics, where each party has access to a local heat bath possibly with a different temperature. Specifically, we have developed a resource theory of local Gaussian work extraction assisted with linear optics (linear interferometers). In doing so, we first introduced a set $\fs$ of free states, namely, products of thermal states (possibly at different temperatures) acted upon by an arbitrary linear interferometer. The states which are not free are then deemed to have a resource called local activity.  We introduced the relative entropy of local activity $A_l$ as a resource monotone and calculated it explicitly for various exemplary cases. In particular, we obtained a closed-form formula for the relative entropy of local activity for arbitrary two-mode Gaussian  states. We also noted that the free states in $\fs$ are locally thermal; hence they are inactive for local Gaussian work extraction (however, there exist nonfree states that are locally thermal too, such as the two-mode squeezed vacuum state). 
Then, we turned to a resource theory of local Gaussian work extraction assisted with linear optics, based on covariance matrices.
The local Gaussian extractable work assisted with linear optics $W_l$ was shown to be a resource monotone, whose properties were discussed in detail.
Finally, we showed examples (both for local activity and local Gaussian extractable work) where we start with two copies of a resource state and apply free operations to get a single state with more resources. These examples are reminiscent of resource distillation protocols in quantum information theory. 

Our results generalize and further advance the resource theory of quantum thermodynamics in a multipartite setting, and could be extended in several directions. Here, we have successfully characterized the set of states $\IW$ whose covariance matrices are such that no local Gaussian work can be extracted, assisted with linear optics. However, it would be very interesting to characterize the set of all states from which no local work extraction is possible, \textit{i.e.}, characterization of the set of quantum states which are locally passive. This is a special case of the so-called quantum marginal problem \cite{Higuchi2003b, Higuchi2003, Bravyi2004, Rudolph2004, Parthasarathy2005, Han2005, Klyachko2006, Eisert2008} which aims at finding global quantum states such that the marginals are fixed and are given. Therefore, the investigation of this problem will shed light on the different possible versions of the quantum marginal problem and solutions thereof. Furthermore, we note that in the rapidly growing field of quantum thermodynamics, mainly in the resource theory of quantum thermodynamics, the discussions on postselection onto free states are  rather scarce compared to corresponding theories of entanglement and coherence. In this work, we envisaged the possibility of postselection onto thermal states; however, this seemed an involved problem and requires further dedicated exposition on its own. Finally, this work is highly relevant to all optical heat engines such as in Ref. \cite{Dechant2014} and we hope that it will lead to further developments in this area in the new setting of resource theories.


\begin{acknowledgements} 
We are grateful to Gerardo Adesso, Ludovico Lami, and Andreas Winter for useful comments on this work. This work was supported by the F.R.S.-FNRS Foundation under Project No. T.0224.18. U.S. acknowledges a fellowship from Universit\'e libre de Bruxelles, M.G.J. a fellowship from the Wiener-Anspach Foundation, and Z.V.H. a fellowship from the FRIA Foundation.
\end{acknowledgements}

\appendix

\section{Covariance matrix for $N$ mode free states}
\label{append:free-states}

In this section, we will show by mathematical induction that the form of the covariance matrix of an $N$-mode free state is given by
\begin{align}
\label{free-form}
 \pmb\Gamma(N)
 =\begin{pmatrix}
 a_{11}\mathbb{I}_2 &R_{12} &\cdots & R_{1N}  \\
 R_{12}^T &a_{22} \mathbb{I}_2&\cdots  & R_{2N}\\
 \vdots &\vdots &\ddots & \vdots \\
 R_{1N}^T  & R_{2N}^T&\cdots  & a_{NN} \mathbb{I}_2
 \end{pmatrix},
 \end{align}
where $R_{ij}$ are $2\times 2$ matrices such that $R_{ij}R_{ij}^T\propto \mathbb{I}_2$ and $R_{ij}\pmb{\omega} R_{ij}^T\propto \pmb{\omega}$. Moreover, all the constants and proportionality constants are such that the covariance matrix is physical.
First, we show that the covariance matrix of any two-mode free state can always be written as 
\begin{align}
\label{free-cov-2}
 \pmb\Gamma(2)
 =\begin{pmatrix}
 a_{11}\mathbb{I}_2 &R_{12}  \\
 R_{12}^T  & a_{22} \mathbb{I}_2
 \end{pmatrix},
 \end{align}
 where $R_{12}$ is such that $R_{12}R_{12}^T\propto \mathbb{I}_2$ and $R_{12}\pmb{\omega} R_{12}^T\propto \pmb{\omega}$. By definition, an arbitrary two-mode free state has covariance matrix $\pmb\Gamma(2)=\mathcal{S} \left(b_{11}\mathbb{I}_2\oplus b_{22}\mathbb{I}_2\right) \mathcal{S}^T$, where $b_{ii}$ $(i=1,2)$ correspond to local temperatures of initial thermal states and $\mathcal{S}$ is an orthogonal symplectic transformation which is combination of a single beam splitter and at least three phase shifters. Thus, for two modes, $\mathcal{S}$ is given by
\begin{align*}
\mathcal{S} &=\begin{pmatrix}
R_1 & 0\\
0 & R_2
\end{pmatrix} 
\begin{pmatrix}
\cos\theta \mathbb{I}_2 & \sin\theta \mathbb{I}_2\\
-\sin\theta \mathbb{I}_2& \cos\theta \mathbb{I}_2
\end{pmatrix} \begin{pmatrix}
R_3 & 0\\
0 & R_4
\end{pmatrix} \\
&=
\begin{pmatrix}
\cos\theta R_1 R_3& \sin\theta R_1R_4\\
-\sin\theta R_2 R_3& \cos\theta R_2R_4
\end{pmatrix}.
\end{align*} 
Now, it is easy to see that $\pmb\Gamma(2)$ has a  similar form to Eq. \eqref{free-cov-2}.

Let us assume that any $N$-mode covariance matrix $\pmb\Gamma(N)$ of any free state  state can be written as Eq. \eqref{free-form}. Now, we bring another mode in a thermal state, and the $(N+1)$-mode free state (as adding ancillae in thermal states is freely allowed) is written as follows.
\begin{align}
 &\pmb\Gamma(N+1)\nonumber\\
 &=\begin{pmatrix}
 \begin{array}{c c c|c c}
 a_{11}\mathbb{I}_2  &\cdots &R_{1(N-1)} & R_{1N} & 0 \\
 \vdots  &\ddots &\vdots & \vdots & \vdots \\
 R_{1(N-1)}^T &\cdots &a_{(N-1)(N-1)}\mathbb{I}_2 & R_{(N-1)N} & 0\\ \hline
 R_{1N}^T  &\cdots & R_{(N-1)N}^T & a_{NN} \mathbb{I}_2 & 0\\
 0   &\cdots & 0 & 0&b \mathbb{I}_2
 \end{array}
 \end{pmatrix}.
 \end{align}
We next apply the linear interferometric transformation on this covariance matrix, which can be factorized into $N$ beam splitter and phase shifter transformations acting sequentially between the $N$th and $(N+1)$th, $(N-1)$th and $(N+1)$th, up to $1$st and $(N+1)$th modes. This provides the most general $N+1$ free state. The beam splitter and phase shifters on $N$th and $(N+1)$th modes correspond to following orthogonal symplectic matrix:  
 \begin{align}
 \mathcal{S}
 =\begin{pmatrix}
 \begin{array}{c c c|c c}
 \mathbb{I}_2 &\cdots &0 & 0 & 0 \\
 \vdots  &\ddots &\vdots & \vdots & \vdots \\
 0  &\cdots &\mathbb{I}_2 & 0 & 0\\ \hline
 0  &\cdots & 0 & \cos\theta T_1 & \sin\theta T_1T_4\\
 0   &\cdots & 0 & -\sin\theta T_2& \cos\theta T_2T_4
 \end{array}
 \end{pmatrix},
 \end{align}
 where $T_1$, $T_2$ and $T_4$ are rotation matrices. This yields

\begin{widetext}
 \begin{align*}
 &\mathcal{S} \pmb\Gamma(N+1)\mathcal{S}^T
 \\
 &=\begin{pmatrix}
 \begin{array}{c c c c|c c}
 a_{11}\mathbb{I}_2 &R_{12} &\cdots &R_{1(N-1)} & \cos\theta  R_{1N} T_1^T & -\sin\theta R_{1N} T_2^T \\
 R_{12}^T &a_{22} \mathbb{I}_2&\cdots &R_{2(N-1)} & \cos\theta  R_{2N} T_1^T & -\sin\theta R_{2N} T_2^T\\
 \vdots &\vdots &\ddots &\vdots & \vdots & \vdots \\
 R_{1(N-1)}^T &R_{2(N-1)}^T &\cdots &a_{(N-1)(N-1)}\mathbb{I}_2 & \cos\theta R_{(N-1)N}T_1^T & -\sin\theta R_{(N-1)N} T_2^T\\ \hline
 \cos\theta T_1 R_{1N}^T  & \cos\theta T_1 R_{2N}^T&\cdots & \cos\theta T_1 R_{(N-1)N}^T & (a_{NN}\cos^2\theta + b\sin^2\theta)\mathbb{I}_2 & (b-a_{NN})\sin\theta\cos\theta T_1 T_2^T\\
 -\sin\theta T_2R_{1N}^T  & -\sin\theta T_2R_{2N}^T &\cdots & -\sin\theta T_2R_{(N-1)N}^T & (b-a_{NN})\sin\theta\cos\theta T_2 T_1^T  &(a_{NN}\sin^2\theta + b\cos^2\theta)\mathbb{I}_2
 \end{array}
 \end{pmatrix}.
 \end{align*}
\end{widetext}
The above matrix is of the form given by Eq. \eqref{free-form}. Now, we can apply the beam splitter and phase shifters between $(N-1)$th and $(N+1)$th mode and so on. It can be seen by mere inspection that the application of the beam splitter and phase shifters always yields the covariance matrix compatible with Eq. \eqref{free-form}. Thus, after application of all beam splitters and phase shifters, we have
\begin{align}
 \pmb\Gamma(N+1)
 =\begin{pmatrix}
 a_{11}\mathbb{I}_2 &R_{12} &\cdots & R_{1(N+1)}  \\
 R_{12}^T &a_{22} \mathbb{I}_2&\cdots  & R_{2(N+1)}\\
 \vdots &\vdots &\ddots & \vdots \\
 R_{1(N+1)}^T  & R_{2(N+1)}^T&\cdots  & a_{(N+1)(N+1)} \mathbb{I}_2
 \end{pmatrix}.
 \end{align}
Therefore, from mathematical induction, we prove that the form of covariance matrices of free states is given by Eq. \eqref{free-form}.

\section{Convexity of the covariance matrices of the set of free states}
\label{append:non-convex}

We have already noted that the set of free states is not convex as the set of Gaussian states is not convex. However, we show here that the covariance matrices of the set of free states form a convex set. Let us first consider a very simple example of two thermal states corresponding to two different modes with covariance matrices $\Gamma_1=a\mathbb{I}$ and $\Gamma_2= b\mathbb{I}$. After the beam-splitter transformation $\mathcal{S}_{BS}(\eta)$ of transmittivity $\eta$ on both modes, we have
\begin{align*}
&\pmb\Gamma'=\mathcal{S}_{BS}(\eta)\Gamma_1\oplus \Gamma_2 \mathcal{S}_{BS}^T(\eta)\\
&=\left(\begin{array}{cc} (a\eta +b(1-\eta))\mathbb{I}_2 & (b-a)\sqrt{\eta(1-\eta)}\mathbb{I}_2\\  (b-a)\sqrt{\eta(1-\eta)}\mathbb{I}_2 &(a(1-\eta) +b\eta)\mathbb{I}_2\end{array}\right).
\end{align*}
Similarly, 
\begin{align*}
\pmb\Gamma''&=\frac{1}{2}\mathcal{S}_{BS}(\eta) \Gamma_1\oplus\Gamma_2 \mathcal{S}_{BS}^T(\eta) + \frac{1}{2} \mathcal{S}_{BS}(\tau) \Gamma_1\oplus \Gamma_2 \mathcal{S}_{BS}^T(\tau)\\
&=\left(\begin{array}{cc} \alpha\mathbb{I}_2 & \gamma\mathbb{I}_2\\  \gamma\mathbb{I}_2 &\beta\mathbb{I}_2\end{array}\right),
\end{align*}
where 
\begin{align*}
&\alpha=a\frac{(\eta+\tau)}{2} +b\left(1-\frac{(\eta+\tau)}{2}\right),\\
&\beta=a\left(1-\frac{(\eta+\tau)}{2}\right) +b\frac{(\eta+\tau)}{2},~\mathrm{and}\\
&\gamma=\frac{(b-a)}{2}\left(\sqrt{\eta(1-\eta)}+\sqrt{\tau(1-\tau)}\right).
\end{align*}
The above covariance matrix can be written as 
\begin{align*}
\pmb\Gamma''&=\mathcal{S}_{BS}(\mu) \left(\begin{array}{cc} c\mathbb{I}_2 & 0\\  0 & d\mathbb{I}_2\end{array}\right) \mathcal{S}_{BS}^T(\mu),
\end{align*}
where
\begin{align*}
&c=\frac{1}{2}\left[ (\alpha+\beta) +\sqrt{(\beta-\alpha)^2+4\gamma^2} \right],\\
&d=\frac{1}{2}\left[ (\alpha+\beta) -\sqrt{(\beta-\alpha)^2+4\gamma^2} \right],\\
&\mu=\cos^2{\theta}~\mathrm{and}~ \theta=\frac{1}{2}\arctan\left(\frac{2\gamma}{\beta-\alpha}\right).
\end{align*}
This shows that $\pmb\Gamma''$ is a free covariance matrix.
\begin{prop}
The free covariance matrices form a convex set.
\end{prop}
\begin{proof}
Let us consider a set of free covariance matrices $\{\pmb\Gamma_{free}^{(j)}\}_{j=1}^M$. Then, to prove the proposition we need to show that for any probability distribution $\{p_j\}_{j=1}^M$, $\sum_j p_j \pmb\Gamma_{free}^{(j)}$ is also a free covariance matrix. Further, it suffices to prove this for convex combination of two free covariance matrices of two modes only. Then using Eq. \eqref{eq:free-cov}, consider two two-mode free covariance matrices
\begin{align*}
\pmb\Gamma_{free}^{(1)}
 =\begin{pmatrix}
 a_{1}\mathbb{I}_2 &R_1  \\
 R_1^T  & a_{2} \mathbb{I}_2
 \end{pmatrix},~\mathrm{and}~\pmb\Gamma_{free}^{(2)}
 =\begin{pmatrix}
 b_{1}\mathbb{I}_2 &R_2  \\
 R_2^T  & b_{2} \mathbb{I}_2
 \end{pmatrix}.
\end{align*}
Now,
\begin{align*}
p_1\pmb\Gamma_{free}^{(1)}+p_2\pmb\Gamma_{free}^{(2)}
 =\begin{pmatrix}
 (p_1 a_{1}+p_2 b_{1})\mathbb{I}_2 & \tilde{R}  \\
 \tilde{R}^T & p_1(a_{2}+ p_2b_{2})\mathbb{I}_2
 \end{pmatrix},
\end{align*}
where $\tilde{R}=p_1 R_1 + p_2 R_2 $. For the above to be a free covariance matrix, we need to show that $\tilde{R}\tilde{R}^T\propto \mathbb{I}_2$ and $\tilde{R}\pmb\omega \tilde{R}^T\propto \pmb\omega$. We see that
\begin{align*}
\tilde{R}\tilde{R}^T &\propto p_1^2 \mathbb{I}_2  + p_1p_2(R_1R_2^T+R_2R_1^T)+ p_2^2 \mathbb{I}_2\propto \mathbb{I}_2.
\end{align*}
The last proportionality follows from the fact that for two real matrices $\{R_i\}_{i=1}^2$ such that $R_i R_i^T\propto \mathbb{I}_2$, $R_1R_2^T+R_2R_1^T$ is also proportional to identity. Further,
\begin{align*}
\tilde{R}\pmb\omega\tilde{R}^T &\propto p_1^2 \pmb\omega + p_2^2 \pmb\omega + p_1p_2(R_1\pmb\omega R_2^T+R_2 \pmb\omega R_1^T)\propto \pmb\omega.
\end{align*}
The last proportionality again follows from the fact that for two real matrices $\{R_i\}_{i=1}^2$ such that $R_i \pmb\omega R_i^T\propto \pmb\omega$ and $R_i R_i^T\propto \mathbb{I}_2$ then $R_1 \pmb\omega R_2^T+R_2 \pmb\omega R_1^T$ is also proportional $\pmb\omega$. This completes the proof of the proposition.
\end{proof}

\begin{rem}
The convexity of the set of free covariance matrices seems to be a general property of Gaussian free states under certain assumptions \cite{Lami2018}, which we choose not to discuss in the present work.
\end{rem}

\section{Postselecting onto free states}
\label{append:postselection}
In order to prove that postselection onto a free state is a free operation, we use the phase-space picture. We start by considering the case where one postselects only one mode, say, the last one. Suppose that we are given a free state $\pmb\Gamma(N)$ as in Eq. \eqref{eq:free-cov} which can be rewritten as
 \begin{align*}
 \pmb\Gamma(N)=\begin{pmatrix}
 \pmb\Gamma_{AA} & \pmb\Gamma_{AB}\\
 \pmb\Gamma_{AB}^T & \pmb\Gamma_{BB} 
 \end{pmatrix}.
 \end{align*}
 Here $\pmb\Gamma_{AA}$ is a $2(N-1)\times 2(N-1)$ covariance matrix, $\pmb\Gamma_{BB}$ is $2\times 2$ covariance matrix which is proportional to identity, and $\pmb\Gamma_{AB}$ is a $2(N-1)\times 2$ matrix. The structure of these matrices can be read from Eq. \eqref{eq:free-cov}. Now, we want to postselect the $N$th mode in the covariance matrix $\pmb\gamma = (\bar{n}'+1/2)\mathbb{I}_2$. Such a postselection results in the remaining $(N-1)$ modes covariance matrix, which is given by
 \begin{align*}
  \tilde{\pmb\Gamma}(N-1) &=  \pmb\Gamma_{AA} - \pmb\Gamma_{AB} (\pmb\Gamma_{BB} + \pmb\gamma)^{-1}\pmb\Gamma_{AB}^T\nonumber\\
  &= \pmb\Gamma''(N)/(\pmb\Gamma_{BB} + \pmb\gamma),
 \end{align*}
 where $\pmb\Gamma''(N)/(\pmb\Gamma_{BB} + \pmb\gamma)$ denotes the Schur complement of $(\pmb\Gamma_{BB} + \pmb\gamma)$ in $\pmb\Gamma''(N)$ and 
 \begin{align*}
 \pmb\Gamma''(N)=\begin{pmatrix}
 \pmb\Gamma_{AA} & \pmb\Gamma_{AB}\\
 \pmb\Gamma_{AB}^T & \pmb\Gamma_{BB} +\pmb\gamma
 \end{pmatrix}.
 \end{align*}
 Now, using Cauchy's interlacing theorem for the eigenvalues of Schur complements \cite{Zhang2005} and eigenvalues of symmetric matrices, we have 
 \begin{align*}
 \lambda_{\mathrm{min}} \left(  \pmb\Gamma''(N)/(\pmb\Gamma_{BB} + \pmb\gamma) \right) &\geq \lambda_{\mathrm{min}} \left(  \pmb\Gamma''(N)\right)\\
 &=\lambda_{\mathrm{min}} \left(  \pmb\Gamma_{BB} + \pmb\gamma \right)\\
 &\geq \lambda_{\mathrm{min}} \left(  \pmb\Gamma_{BB} \right) + \lambda_{\mathrm{min}} \left(  \pmb\gamma \right)\\
 &\geq 1.
 \end{align*}
 In the above, we have used the fact that $\lambda_{\mathrm{min}}(A+B)\geq \lambda_{\mathrm{min}}(A)+ \lambda_{\mathrm{min}}(B)$ which follows from the Courant-Fischer-Weyl min-max theorem \cite{Zhang2005}. Thus, $\lambda_{\mathrm{min}} (\tilde{\pmb\Gamma}(N-1)) \geq 1$. Now, notice that $\pmb\Gamma_{BB} + \pmb\gamma$ is proportional to $\mathbb{I}_2$ and let $\pmb\Gamma_{BB} + \pmb\gamma = x \mathbb{I}_2$ with $x\geq 1$. Therefore, 
 \begin{align*}
  \tilde{\pmb\Gamma}(N-1) &=  \pmb\Gamma_{AA} - x^{-1}\pmb\Gamma_{AB} \pmb\Gamma_{AB}^T.
 \end{align*}
 Notice from Eq. \eqref{eq:free-cov} that
  \begin{align*}
  \pmb\Gamma_{AB} =\begin{pmatrix}
  R_{1N}\\
  \vdots\\
  R_{(N-1)N}
  \end{pmatrix}.
 \end{align*}
 Therefore,
  \begin{align*}
  \pmb\Gamma_{AB} \pmb\Gamma_{AB}^T &=\begin{pmatrix}
  R_{1N} \\
  \vdots\\
  R_{(N-1)N}
  \end{pmatrix} \begin{pmatrix}
  R_{1N}^T &\cdots & R_{(N-1)N}^T 
  \end{pmatrix}\nonumber\\
   &=\begin{pmatrix}
 c_{11}\mathbb{I}_2 &R'_{12} &\cdots & R'_{1(N-1)} \\
 R_{12}^{'T} &c_{22} \mathbb{I}_2&\cdots & R'_{2(N-1)} \\
 \vdots &\vdots &\ddots &\vdots\\
 R_{1(N-1)}^{'T}  & R^{'T}_{2(N-1)} &\cdots &  c_{(N-1)(N-1)} \mathbb{I}_2
 \end{pmatrix},
 \end{align*}
 where, $c_{ii} = R_{iN} R_{iN}^T$ and $R'_{ij}=R_{iN} R_{jN}^T$ for $i<j$. Also, $R'_{ij}R_{ij}^{'T}\propto\mathbb{I}_2$ and $R'_{ij}\pmb{\omega}R_{ij}^{'T}\propto \pmb{\omega}$. Thus, $ \tilde{\pmb\Gamma}(N-1)$ has the following form:
 \begin{align*}
  \tilde{\pmb\Gamma}(N-1) &=\begin{pmatrix}
 d_{11}\mathbb{I}_2 &R''_{12} &\cdots & R''_{1(N-1)} \\
 R_{12}^{''T} &d_{22} \mathbb{I}_2&\cdots & R''_{2(N-1)} \\
 \vdots &\vdots &\ddots &\vdots\\
 R_{1(N-1)}^{''T}  & R^{''T}_{2(N-1)} &\cdots &  d_{(N-1)(N-1)} \mathbb{I}_2
 \end{pmatrix},
  \end{align*}
  where $R''_{ij}R_{ij}^{''T}\propto\mathbb{I}_2$ and $R''_{ij}\pmb{\omega}R_{ij}^{''T}\propto\pmb{\omega}$. We have also used the fact that given two matrices $T_1$ and $T_2$ such that $T_iT_i^T\propto \mathbb{I}_2$ and $T_i\pmb\omega T_i^T\propto \pmb\omega$ $(i=1,2)$, we have $T'=(T_1-T_2)$ such that $T' T^{'T}\propto \mathbb{I}_2$ and $T'\pmb\omega T^{'T}\propto \pmb\omega$. The above matrix $\tilde{\pmb\Gamma}(N-1)$ has a similar form to Eq. \eqref{eq:free-cov}. Thus, we have proved that postselection onto a single-mode free state leaves the remaining state free (see Fig.~\ref{fig:postselecting-free-state}). Similar arguments can be presented to show that postselection onto any free state leaves the remaining state free and, therefore, postselection onto a free state is a free operation.

\begin{figure}
\centering
\includegraphics[width=55 mm]{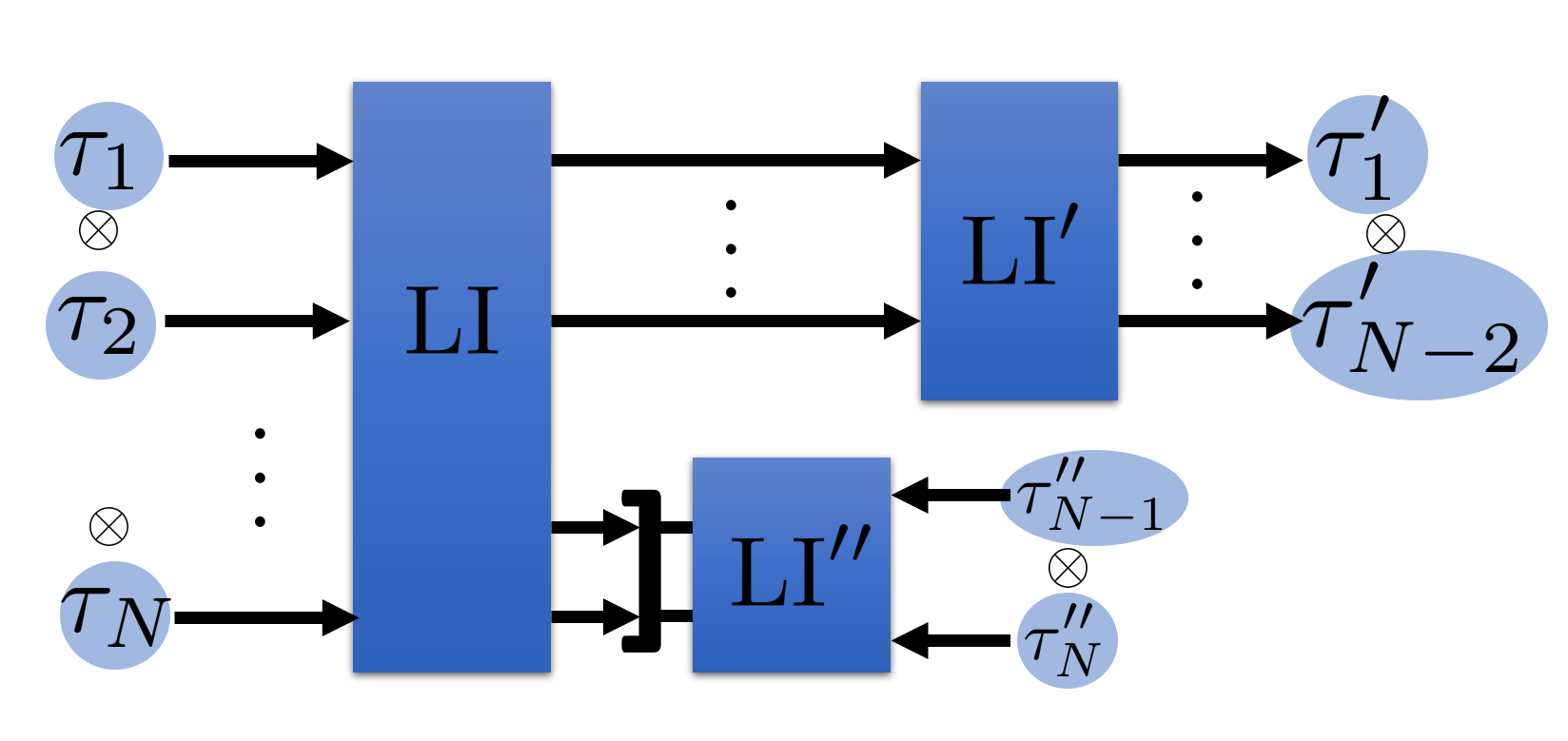}
\caption{The figure shows that the postselection of the last two modes of an $N$-mode free state $\rho^f_{1\cdots N}$ onto a free state results in another free state.}
\label{fig:postselecting-free-state}
\end{figure}

\begin{prop}
For a free channel as given in Eq. \eqref{eq:free-channel} such that for all its Kraus operators $K_i$, $K_i\sigma K_i^\dagger$ is a free state if $\sigma\in\fs$, then
\begin{align*}
\act\left(\rho\right) &\geq \sum_i p_i \act\left(\rho_i\right),
\end{align*}
where $\rho_i=\left(K_i\rho K_i^\dagger\right)/p_i$ and $p_i=\mathrm{Tr}\left(K_i\rho K_i^\dagger\right)$. This property implies monotonicity of the relative entropy of local activity on an average under selective free operations.
\end{prop}
\begin{proof}
The proof follows from the property of the relative entropy \cite{Vedral1998}, which states that
\begin{align}
\label{eq:prop-rel}
\infdiv*{\rho}{\sigma}\geq \sum_{i}p_i \infdiv*{\rho_i}{\frac{K_i\sigma K_i^\dagger}{q_i}},
\end{align}
where $q_i=\mathrm{Tr}\left(K_i\sigma K_i^\dagger\right)$. From Eq. \eqref{eq:prop-rel}, we have
\begin{align*}
\act\left(\rho\right)&= \infdiv*{\rho}{\sigma^*}\\
&\geq \sum_{i}p_i \infdiv*{\rho_i}{\frac{K_i\sigma^* K_i^\dagger}{q_i}}\\
&= \sum_{i}p_i \infdiv*{\rho_i}{\tilde{\sigma}_i}\\
&\geq \sum_{i}p_i \min_{\tilde{\sigma}_i\in\fs}\infdiv*{\rho_i}{\tilde{\sigma}_i}= \sum_{i}p_i \act\left(\rho_i\right).
\end{align*}
This completes the proof of the proposition.
\end{proof}
Now, we address the question of the existence of free channels such that all of their Kraus operators are free; \textit{i.e.}, $K_i\sigma K_i^\dagger$ is a free state if $\sigma\in\fs$ for all $i$. We show that an arbitrary single-mode free channel (Eq. \eqref{eq:free-channel}) does not admit a Kraus decomposition such that all its Kraus operators are free. To this end, consider a single-mode free channel which, without loss of generality, can be defined as follows:
\begin{align}
\label{eq:thermal-loss-channel}
\Phi_{LT}(\rho^S) = \mathrm{Tr}_{A}\left[ U_{SA}^{BS} \left(\rho^S\otimes\tau_{th}^{A}\right) U_{SA}^{BS\dagger}\right],
\end{align}
where $U_{SA}^{BS}$ is a beam-splitter unitary and 
\begin{align*}
\tau_{th}^{A}&=\sum_{n=0}^\infty p_n\ket{n}\bra{n}=\sum_{n=0}^\infty \frac{\bar{n}_{\tau}^n}{(\bar{n}_{\tau}+1)^{(n+1)}}\ket{n}\bra{n}\\
&=(1-x)\sum_{n=0}^\infty x^n\ket{n}\bra{n},
\end{align*}
where $x=\bar{n}_{\tau}/(\bar{n}_{\tau}+1)$ and $\bar{n}_\tau$ is the average number of photons in state $\tau_{th}$. Now,
\begin{align*}
\Phi_{LT}(\rho^S) &= \sum_{m,n=0}^{\infty} \left(\sqrt{p_n}\bra{m} U_{SA}^{BS}\ket{n}\right) \rho^S\left(\sqrt{p_n}\bra{n} U_{SA}^{BS\dagger}\ket{m}\right)\\
&= \sum_{m,n=0}^{\infty} K_{mn} \rho^S K_{mn}^{\dagger},
\end{align*}
where $K_{mn}= \sqrt{p_n}\bra{m} U_{SA}^{BS}\ket{n}$. The expression for these Kraus operators can be found readily using Ref. \cite{Ivan2011}. In particular, $K_{mn} = \sum_{m_1,n_1=0}^{\infty}\sqrt{p_n}\bra{m_1,m} U_{SA}^{BS}\ket{n_1,n}\ket{m_1}\bra{n_1}$ and 
\begin{align*}
&\bra{m_1,m} U_{SA}^{BS}\ket{n_1,n} \\&=\sqrt{\frac{m_1!\, m!}{n_1!\, n!}}\sum_{s=0}^{n_1}\sum_{t=0}^{n}\binom{n_1}{s} \binom{n}{t} (-1)^{n-t}\\
&\times \eta^{n_1-s+t}(\sqrt{1-\eta^2})^{s+n-t}\delta_{m,s+t}\delta_{m_1,n_1+n-s-t},
\end{align*}
where $\eta$ is the transmittance of the beam splitter. Now, let us consider the Kraus operator $K_{00}$; this is given by
\begin{align*}
K_{00}
&= \sum_{m_1=0}^{\infty} \sqrt{p_0}\, \eta^{m_1}\ket{m_1}\bra{m_1}.
\end{align*}
Let us consider that $\rho^S=(1-y)\sum_{n=0}^\infty y^n\ket{n}\bra{n}$ is a thermal state with $y=\frac{\bar{n}_\rho}{\bar{n}_\rho+1}$. Then,
\begin{align*}
\rho^S_{00}\equiv\frac{K_{00}\rho^SK_{00}^\dagger}{\mathrm{Tr}[K_{00}\rho^SK_{00}^\dagger]}
&= (1-y\eta^2) \sum_{n=0}^{\infty} {(y \eta^{2})}^n\ket{n}\bra{n}\\
&= (1-z)\sum_{n=0}^{\infty} {z}^n\ket{n}\bra{n},
\end{align*}
where $z=y\eta^2$. Therefore, $\rho^S_{00}$ is indeed a thermal state. However, $K_{00}$ is the only Kraus operator that maps a thermal state to another thermal state. Let us consider another Kraus operator $K_{10}$, which is given by
\begin{align*}
K_{10}
&= \sum_{m_1=0}^{\infty} \sqrt{p_0}\,\sqrt{m_1+1} \eta^{m_1} \sqrt{1-\eta^2}\ket{m_1}\bra{m_1+1}.
\end{align*}
Now,
\begin{align*}
\rho^S_{10}&\equiv\frac{K_{10}\rho^SK_{10}^\dagger}{\mathrm{Tr}[K_{10}\rho^SK_{10}^\dagger]}\\
&= \frac{\sum_{n=0}^{\infty}(n+1)\eta^{2n}y^{n+1}\ket{n}\bra{n}}{\sum_{n=0}^{\infty}(n+1)\eta^{2n}y^{n+1}}\\
&= (1-y\eta^2)\sum_{n=0}^{\infty}(n+1)(y\eta^2)^n\ket{n}\bra{n}
\end{align*}
This implies that $\rho^S_{10}$ is not a thermal state. Therefore, for the one-mode case, all the Kraus operators corresponding to the  most general free channel are not free.

\section{Relative entropy of local activity for two modes}
\label{append:mono-two}
In this section, we calculate explicitly the relative entropy of local activity for any arbitrary two-mode Gaussian state. Consider a Gaussian state $\rho_1({\bf d},\pmb\Gamma_1)$ of two modes with covariance matrix
\begin{align*}
\pmb\Gamma_1=\begin{pmatrix}
A & C\\
C^T& B
\end{pmatrix},
\end{align*}
and the displacement vector ${\bf d}=\begin{pmatrix}
d_1&d_2&d_3&d_4
\end{pmatrix}^T$.
Consider a generic free state $\rho_2({\bf 0},\pmb\Gamma_2) = U^{\mathrm{PG}}(\tau_1\otimes\tau_2)U^{\mathrm{PG}\dagger}$ such that $\pmb\Gamma_2=\mathcal{S} \left(b_1\mathbb{I}_2\oplus b_2\mathbb{I}_2\right) \mathcal{S}^T$, where $\mathcal{S}$ is an orthogonal symplectic matrix and is given by
\begin{align*}
\mathcal{S} &=\begin{pmatrix}
R_1 & 0\\
0 & R_2
\end{pmatrix} 
\begin{pmatrix}
\cos\theta \mathbb{I}_2 & \sin\theta \mathbb{I}_2\\
-\sin\theta \mathbb{I}_2& \cos\theta \mathbb{I}_2
\end{pmatrix} \begin{pmatrix}
R_3 & 0\\
0 & R_4
\end{pmatrix} \\
&=
\begin{pmatrix}
\cos\theta R_1 R_3& \sin\theta R_1R_4\\
-\sin\theta R_2 R_3& \cos\theta R_2R_4
\end{pmatrix}.
\end{align*} 
Also, the symplectic form is given by
\begin{align*}
\pmb\Omega=\pmb\omega\oplus \pmb\omega~\mathrm{and}~\pmb\omega = \begin{pmatrix}
0 & 1\\
-1 & 0
\end{pmatrix}. 
\end{align*}

\begin{widetext}

Now, the relative entropy between $\rho_1({\bf d},\pmb\Gamma_1)$ and $\rho_2({\bf 0},\pmb\Gamma_2)$ is given by (\cite{Pirandola2017}; also see Sec. \ref{sec:prelims})
\begin{align*}
&\infdiv*{\rho_1({\bf d},\pmb\Gamma_1)}{\rho_2({\bf 0},\pmb\Gamma_2)} = -S(\rho_1) + \frac{1}{2} \ln \mathrm{det}\left( \pmb\Gamma_2 + i \frac{\Omega}{2}\right) +  \frac{1}{2}\mathrm{Tr}\left(\pmb\Gamma_1 {\bf G}_2\right) +\frac{1}{2}{\bf d}^T{\bf G}_2 {\bf d}:=\infdiv*{\pmb\Gamma_1}{\pmb\Gamma_2}, 
\end{align*} 
where 
\begin{align*}
{\bf G}_2 &= - \pmb\Omega \mathcal{S} \left(a_1\mathbb{I}_2\oplus a_2 \mathbb{I}_2\right) \mathcal{S}^T \pmb\Omega\\
&= -  \mathcal{S} \mathcal{S}^T\pmb\Omega \mathcal{S} \left(a_1\mathbb{I}_2\oplus a_2 \mathbb{I}_2\right) \mathcal{S}^T \pmb\Omega \mathcal{S} \mathcal{S}^T
= \mathcal{S} \left(a_1\mathbb{I}_2\oplus a_2 \mathbb{I}_2\right) \mathcal{S}^T.
\end{align*} 
Here $a_i=2\coth^{-1}(2 b_i)$, $(i=1,2)$. Let us first calculate ${\bf G}_2$ and $\pmb\Gamma_1 {\bf G}_2$: 

\begin{align*}
{\bf G}_2 &=
\begin{pmatrix}
\cos\theta R_1R_3 & \sin\theta R_1R_4\\
-\sin\theta R_2R_3 & \cos\theta R_2R_4
\end{pmatrix} 
\begin{pmatrix}
a_1\mathbb{I}_2 & 0\\
0 & a_2\mathbb{I}_2
\end{pmatrix}
 \begin{pmatrix}
\cos\theta R_3^TR_1^T & -\sin\theta R_3^T R_2^T\\
\sin\theta R_4^TR_1^T & \cos\theta R_4^TR_2^T
\end{pmatrix} 
\\
&= \begin{pmatrix}
\left(a_1 \cos^2\theta + a_2 \sin^2\theta\right) \mathbb{I}_2 & (a_2-a_1) \sin\theta \cos\theta R_{12}\\
(a_2-a_1) \sin\theta \cos\theta R_{12}^T& \left(a_1 \sin^2\theta + a_2  \cos^2\theta \right) \mathbb{I}_2 
\end{pmatrix}.
\end{align*} 
\begin{align*}
\pmb\Gamma_1 {\bf G}_2&= \begin{pmatrix}
A & C\\
 C^T&B
\end{pmatrix}\begin{pmatrix}
\left(a_1 \cos^2\theta + a_2 \sin^2\theta\right) \mathbb{I}_2 & (a_2-a_1) \sin\theta \cos\theta R_{12}\\
(a_2-a_1) \sin\theta \cos\theta  R_{12}^T& \left(a_1 \sin^2\theta + a_2  \cos^2\theta \right) \mathbb{I}_2 
\end{pmatrix}\\
&= \begin{pmatrix}
\left(a_1 \cos^2\theta + a_2 \sin^2\theta\right) A +(a_2-a_1) \sin\theta \cos\theta C R_{12}^T& (a_2-a_1) \sin\theta \cos\theta  AR_{12}+\left(a_1 \sin^2\theta + a_2  \cos^2\theta \right)C\\
\left(a_1 \cos^2\theta + a_2 \sin^2\theta\right)C^T+(a_2-a_1) \sin\theta \cos\theta BR_{12}^T&\left(a_1 \sin^2\theta + a_2  \cos^2\theta \right) B +(a_2-a_1) \sin\theta \cos\theta C^T R_{12}
\end{pmatrix},
\end{align*} 
where $R_{12}=R_1R_2^T$ and
\begin{align*}
C R_{12}^T&=\begin{pmatrix}
c_1 & c_2\\
c_3& c_4
\end{pmatrix} \begin{pmatrix}
\cos\delta\phi & -\sin\delta\phi\\
\sin\delta\phi& \cos\delta\phi
\end{pmatrix}=\begin{pmatrix}
c_1\cos\delta\phi +c_2 \sin\delta\phi& -c_1\sin\delta\phi+c_2\cos\delta\phi\\
c_3\cos\delta\phi+c_4\sin\delta\phi& -c_3\sin\delta\phi+c_4\cos\delta\phi
\end{pmatrix}.
\end{align*}
Thus, denoting $\mathrm{Tr}[C]=c$ and $c_2-c_3=\upsilon$, we have $\mathrm{Tr}[C R_{12}^T+C^TR_{12}]=2c\cos\delta\phi+2\upsilon\sin\delta\phi$. Therefore,
\begin{align*}
\mathrm{Tr}\left(\pmb\Gamma_1 {\bf G}_2\right)
&=\alpha \left(a_1 \cos^2\theta + a_2 \sin^2\theta\right)  + \beta\left(a_1 \sin^2\theta + a_2  \cos^2\theta \right) + (a_2-a_1) \sin\theta \cos\theta \left( 2c\cos\delta\phi+2\upsilon\sin\delta\phi. \right)\\
&=\frac{1}{2}\left[(a_1+a_2) (\alpha+\beta)+ (a_1-a_2)(\alpha-\beta)\cos2\theta +2 (a_2-a_1) \sin2\theta \left( c\cos\delta\phi+\upsilon\sin\delta\phi \right)\right],
\end{align*} 
where $\alpha= \mathrm{Tr}[A]$ and $\beta = \mathrm{Tr}[B]$. Let us now calculate $\mathrm{det}\left( \pmb\Gamma_2 + i \frac{\pmb\Omega}{2}\right)$:
\begin{align*}
\mathrm{det}\left( \pmb\Gamma_2 + i \frac{\pmb\Omega}{2}\right)  &= \mathrm{det}\left( \mathcal{S} \left(b_1\mathbb{I}_2\oplus b_2\mathbb{I}_2\right) \mathcal{S}^T + i \frac{\pmb\Omega}{2}\right)\nonumber\\
&= \mathrm{det} \left(\left(b_1\mathbb{I}_2\oplus b_2\mathbb{I}_2\right) + i \frac{\pmb\Omega}{2}\right)
= \left(b_1^2 -\frac{1}{4}\right)  \left(b_2^2 -\frac{1}{4}\right).
\end{align*}
Further, we have
\begin{align*}
{\bf d}^T{\bf G}_2 {\bf d}&=
\begin{pmatrix}
d_1&d_2&d_3&d_4
\end{pmatrix}\begin{pmatrix}
\left(a_1 \cos^2\theta + a_2 \sin^2\theta\right) \mathbb{I}_2 & (a_2-a_1) \sin\theta \cos\theta R_{12}\\
(a_2-a_1) \sin\theta \cos\theta R_{12}^T& \left(a_1 \sin^2\theta + a_2  \cos^2\theta \right) \mathbb{I}_2 
\end{pmatrix} \begin{pmatrix}
d_1\\
d_2\\
d_3\\
d_4
\end{pmatrix}\\
&=(a_1\cos^2\theta + a_2\sin^2\theta)(d_1^2+d_2^2) + (a_1\sin^2\theta+a_2\cos^2\theta)(d_3^2+d_4^2)\\
&~~~+(a_2-a_1)\sin2\theta\left((d_1d_3+d_2d_4)\cos\delta\phi +(d_1d_4-d_2d_3)\sin\delta\phi \right)\\
&=\frac{1}{2}\left[ (a_1+a_2)\tilde{d}_1 + (a_1-a_2)\tilde{d}_2\cos2\theta +2(a_2-a_1)\sin2\theta\left(\tilde{d}_3 \cos\delta\phi +\tilde{d}_4\sin\delta\phi \right)\right],
\end{align*}
where $\tilde{d}_1=(d_1^2+d_2^2+d_3^2+d_4^2)$, $\tilde{d}_2=(d_1^2+d_2^2-d_3^2-d_4^2)$, $\tilde{d}_3=(d_1d_3+d_2d_4)$, and $\tilde{d}_4=(d_1d_4-d_2d_3)$.
Thus, writing the relative entropy as a function of CMs to simplify the notation, we have
\begin{align*}
\infdiv*{\pmb\Gamma_1}{\pmb\Gamma_2}&=-S(\rho_1) + \frac{1}{2} \ln \left(b_1^2 -\frac{1}{4}\right) + \frac{1}{2}\ln \left(b_2^2 -\frac{1}{4}\right) +  \frac{1}{4}\left((a_1+a_2)(\alpha+\beta+\tilde{d}_1) + (a_1-a_2)(\alpha-\beta+\tilde{d}_2) \cos2\theta\right)\\
&~~~+  \frac{1}{2}(a_2-a_1) \sin2\theta \left( (c+\tilde{d}_3)\cos\delta\phi+(\upsilon+\tilde{d}_4)\sin\delta\phi\right)\\
&=-S(\rho_1) + \frac{1}{2} \ln \left(b_1^2 -\frac{1}{4}\right) + \frac{1}{2}\ln \left(b_2^2 -\frac{1}{4}\right) +  \frac{1}{4}\left((a_1+a_2)\tilde{\alpha} + (a_1-a_2)\tilde{\beta} \cos2\theta\right)\\
&~~~+  \frac{1}{2}(a_2-a_1) \sin2\theta \left( \tilde{c}\cos\delta\phi+\tilde{\upsilon}\sin\delta\phi\right),
\end{align*} 
where $\tilde{\alpha}=(\alpha+\beta+\tilde{d}_1)$, $\tilde{\beta}=(\alpha-\beta+\tilde{d}_2)$, $\tilde{c}=(c+\tilde{d}_3)$, and $\tilde{\upsilon}=(\upsilon+\tilde{d}_4)$. We now want to minimize $\infdiv*{\pmb\Gamma_1}{\pmb\Gamma_2}$ with respect to $b_1$, $b_2$, $\theta$, and $\delta\phi=\phi_1-\phi_2$ and for that we put first derivatives of $\infdiv*{\pmb\Gamma_1}{\pmb\Gamma_2}$ with respect to $b_1$, $b_2$, $\theta$,  and $\delta\phi$ equal to zero:
\begin{align*}
&\frac{\partial \infdiv*{\pmb\Gamma_1}{\pmb\Gamma_2}}{\partial b_1}=\frac{-4 b_1 +\tilde{\alpha} +\left(\tilde{\beta}\cos2\theta - 2\left( \tilde{c}\cos\delta\phi+\tilde{\upsilon}\sin\delta\phi \right)\sin2\theta\right)}{1-4 b_1^2}=0;\\
&\frac{\partial \infdiv*{\pmb\Gamma_1}{\pmb\Gamma_2}}{\partial b_2}=\frac{-4 b_2 +\tilde{\alpha} -\left(\tilde{\beta}\cos2\theta - 2\left( \tilde{c}\cos\delta\phi+\tilde{\upsilon}\sin\delta\phi \right)\sin2\theta\right)}{1-4 b_2^2}=0;\\
&\frac{\partial \infdiv*{\pmb\Gamma_1}{\pmb\Gamma_2}}{\partial \theta}=\frac{1}{2}(a_2-a_1)\left(\tilde{\beta}\sin2\theta  + 2\left( \tilde{c}\cos\delta\phi+\tilde{\upsilon}\sin\delta\phi\right)\cos2\theta\right )=0;\\
&\frac{\partial \infdiv*{\pmb\Gamma_1}{\pmb\Gamma_2}}{\partial \delta\phi}= \frac{1}{2}(a_2-a_1) \sin2\theta(-\tilde{c}\sin\delta\phi+\tilde{\upsilon}\cos\delta\phi)=0.
\end{align*} 
To solve the above equation, we consider the following cases.

{\bf Case(1).} $\tan\delta\phi=\tilde{\upsilon}/\tilde{c}$, thus, $\cos\delta\phi=\tilde{c}/\sqrt{\tilde{c}^2+\tilde{\upsilon}^2}$ and $\sin\delta\phi=\tilde{\upsilon}/\sqrt{\tilde{c}^2+\tilde{\upsilon}^2}$. In this case, $\tan2\theta=-2\sqrt{\tilde{c}^2+\tilde{\upsilon}^2}/\tilde{\beta}$. This means $\cos2\theta =\frac{\tilde{\beta}}{\sqrt{\tilde{\beta}^2+4(\tilde{c}^2+\tilde{\upsilon}^2)}}$ and $\sin2\theta=-\frac{2\sqrt{\tilde{c}^2+\tilde{\upsilon}^2} }{\sqrt{\tilde{\beta}^2+4(\tilde{c}^2+\tilde{\upsilon}^2)}}$. In this case,
\begin{align}
\label{eq:app-b1b2}
&b_1= \frac{1}{4} \left[\tilde{\alpha}+\sqrt{\tilde{\beta}^2+4(\tilde{c}^2+\tilde{\upsilon}^2)}\right],~b_2= \frac{1}{4} \left[\tilde{\alpha}-\sqrt{\tilde{\beta}^2+4(\tilde{c}^2+\tilde{\upsilon}^2)}\right].
\end{align}

{\bf Case(2).} $a_1=a_2$. This implies $b_1=b_2$, \textit{i.e.},
$\tan2\theta=\frac{\tilde{\beta}}{2\left( \tilde{c}\cos\delta\phi+\tilde{\upsilon}\sin\delta\phi \right)}$.
This implies $\cos2\theta= \frac{2\left( \tilde{c}\cos\delta\phi+\tilde{\upsilon}\sin\delta\phi \right)}{\sqrt{4\left( \tilde{c}\cos\delta\phi+\tilde{\upsilon}\sin\delta\phi \right)^2+\tilde{\beta}^2}}$ and $\sin2\theta=\frac{\tilde{\beta}}{ \sqrt{4\left( \tilde{c}\cos\delta\phi+\tilde{\upsilon}\sin\delta\phi \right)^2+\tilde{\beta}^2} }$. Therefore, $b_1= \frac{\tilde{\alpha}}{4} =b_2$.

{\bf Case (3).} $\theta=0$; then $\tan\delta\phi=-\tilde{c}/\tilde{\upsilon}$, $\cos\delta\phi=\tilde{\upsilon}/\sqrt{\tilde{c}^2+\tilde{\upsilon}^2}$, and $\sin\delta\phi=-\tilde{c}/\sqrt{\tilde{c}^2+\tilde{\upsilon}^2}$. Here $b_1=(\tilde{\alpha}+\tilde{\beta})/4$ and $b_2=(\tilde{\alpha}-\tilde{\beta})/4$.

Let us consider now double derivatives of $\infdiv*{\pmb\Gamma_1}{\pmb\Gamma_2}$ and construct the Hessian matrix in order to find the true minimum out of all the above extrema:
\begin{align*}
&\frac{\partial^2 \infdiv*{\pmb\Gamma_1}{\pmb\Gamma_2}}{\partial b_1^2}= \frac{8 b_1\left(-4 b_1 +\tilde{\alpha} +\left(\tilde{\beta}\cos2\theta - 2\left( \tilde{c}\cos\delta\phi+\tilde{\upsilon}\sin\delta\phi \right)\sin2\theta\right)\right) - 4 (1-4 b_1^2)}{\left(1-4 b_1^2\right)^2};\\
&\frac{\partial^2 \infdiv*{\pmb\Gamma_1}{\pmb\Gamma_2}}{\partial b_1 \partial b_2}= 0;~\frac{\partial^2 \infdiv*{\pmb\Gamma_1}{\pmb\Gamma_2}}{\partial b_1 \partial \theta}= -\frac{2}{1-4b_1^2}\left(\tilde{\beta}\sin2\theta  +  2\left(  \tilde{c}\cos\delta\phi+\tilde{\upsilon}\sin\delta\phi \right)\cos2\theta\right );\\
&\frac{\partial^2 \infdiv*{\pmb\Gamma_1}{\pmb\Gamma_2}}{\partial b_1 \partial \delta\phi}= -\frac{2\left( - \tilde{c}\sin\delta\phi+\tilde{\upsilon}\cos\delta\phi \right) \sin2\theta}{1-4 b_1^2};\\
&\frac{\partial^2 \infdiv*{\pmb\Gamma_1}{\pmb\Gamma_2}}{\partial b_2^2}= \frac{8 b_2\left(-4 b_2 +\tilde{\alpha} -\left(\tilde{\beta}\cos2\theta - 2\left(  \tilde{c}\cos\delta\phi+\tilde{\upsilon}\sin\delta\phi \right)\sin2\theta\right)\right) - 4 (1-4 b_2^2)}{\left(1-4 b_2^2\right)^2};\\
&\frac{\partial^2 \infdiv*{\pmb\Gamma_1}{\pmb\Gamma_2}}{\partial b_2 \partial \theta}= \frac{2}{1-4b_2^2}\left(\tilde{\beta}\sin2\theta  +  2\left(  \tilde{c}\cos\delta\phi+\tilde{\upsilon}\sin\delta\phi \right)\cos2\theta\right );~\frac{\partial^2 \infdiv*{\pmb\Gamma_1}{\pmb\Gamma_2}}{\partial b_2 \partial \delta\phi}=  \frac{2\left( - \tilde{c}\sin\delta\phi+\tilde{\upsilon}\cos\delta\phi \right) \sin2\theta}{1-4 b_2^2};\\
&\frac{\partial^2 \infdiv*{\pmb\Gamma_1}{\pmb\Gamma_2}}{\partial \theta^2}= (a_2-a_1)\left(\tilde{\beta}\cos2\theta  - 2\left( \tilde{c}\cos\delta\phi+\tilde{\upsilon}\sin\delta\phi \right)\sin2\theta\right );\\
&\frac{\partial^2 \infdiv*{\pmb\Gamma_1}{\pmb\Gamma_2}}{\partial \theta\partial\delta\phi}=(a_2-a_1)\left( - \tilde{c}\sin\delta\phi+\tilde{\upsilon}\cos\delta\phi \right)\cos2\theta.\\
&\frac{\partial^2 \infdiv*{\pmb\Gamma_1}{\pmb\Gamma_2}}{\partial\delta\phi^2}=-\frac{1}{2}(a_2-a_1)\left(  \tilde{c}\cos\delta\phi+\tilde{\upsilon}\sin\delta\phi \right) \sin2\theta.
\end{align*} 
For case 1, we have
\begin{align*}
&\frac{\partial^2 \infdiv*{\pmb\Gamma_1}{\pmb\Gamma_2}}{\partial b_1^2}=\frac{4}{\left(4 b_1^2-1\right)};~\frac{\partial^2 \infdiv*{\pmb\Gamma_1}{\pmb\Gamma_2}}{\partial b_1 \partial b_2}= 0;~\frac{\partial^2 \infdiv*{\pmb\Gamma_1}{\pmb\Gamma_2}}{\partial b_1 \partial \theta}= 0;\\
&\frac{\partial^2 \infdiv*{\pmb\Gamma_1}{\pmb\Gamma_2}}{\partial b_1 \partial \delta\phi}=0;~\frac{\partial^2 \infdiv*{\pmb\Gamma_1}{\pmb\Gamma_2}}{\partial b_2^2}= \frac{4}{\left(4 b_2^2-1\right)};~\frac{\partial^2 \infdiv*{\pmb\Gamma_1}{\pmb\Gamma_2}}{\partial b_2 \partial \theta}= 0;\\
&\frac{\partial^2 \infdiv*{\pmb\Gamma_1}{\pmb\Gamma_2}}{\partial b_2 \partial \delta\phi}=  0;~\frac{\partial^2 \infdiv*{\pmb\Gamma_1}{\pmb\Gamma_2}}{\partial \theta^2}= (a_2-a_1)\sqrt{\tilde{\beta}^2+4(\tilde{c}^2+\tilde{\upsilon}^2)}\\
&\frac{\partial^2 \infdiv*{\pmb\Gamma_1}{\pmb\Gamma_2}}{\partial \theta\partial\delta\phi}=0;~\frac{\partial^2 \infdiv*{\pmb\Gamma_1}{\pmb\Gamma_2}}{\partial\delta\phi^2}=\frac{(a_2-a_1)(\tilde{c}^2+\tilde{\upsilon}^2) }{\sqrt{\tilde{\beta}^2+4(\tilde{c}^2+\tilde{\upsilon}^2)}}.
\end{align*} 
All eigenvalues of the corresponding Hessian are positive, so this case corresponds to a minimum.
%
%
%
Let the solution correspond to  $(b_1,b_2,\theta,\delta\phi)=(b_1^*,b_2^*,\theta^*,\delta\phi^*)$; then
\begin{align*}
\infdiv*{\pmb\Gamma_1}{\pmb\Gamma_2}&=-S(\rho_1) + \frac{1}{2} \ln \left(b_1^{*2} -\frac{1}{4}\right) + \frac{1}{2}\ln \left(b_2^{*2} -\frac{1}{4}\right) +  \frac{1}{4} (a_1^*+a_2^*)\tilde{\alpha} \\
&~~~+\frac{1}{4}(a_1^*-a_2^*)\left( \tilde{\beta} \cos2\theta^*-  2  \left( \tilde{c}\cos\delta\phi^*+\tilde{\upsilon}\sin\delta\phi^*\right) \sin2\theta^*\right)\\
&=-S(\rho_1) + \frac{1}{2} \ln \left(b_1^{*2} -\frac{1}{4}\right) + \frac{1}{2}\ln \left(b_2^{*2} -\frac{1}{4}\right) +  a_1^*b_1^*+a_2^*b_2^*\\
&=-S(\rho_1) +\frac{1}{2}\sum_{i=1}^2\left[ \ln \left(b_i^* +\frac{1}{2}\right) + \ln \left(b_i^* -\frac{1}{2}\right) \right]+\sum_{i=1}^2\left[b_i^* \ln \left(b_i^* +\frac{1}{2}\right) -b_i^* \ln \left(b_i^* -\frac{1}{2}\right) \right]\\
&=-S(\rho_1) +\sum_{i=1}^2 \left(b_i^* +\frac{1}{2}\right) \ln \left(b_i^* +\frac{1}{2}\right) -  \left(b_i^* -\frac{1}{2}\right)\ln \left(b_i^* -\frac{1}{2}\right)\\
&=\sum_{i=1}^2 \left[g(b_i^*)-g(\nu_i)\right].
\end{align*} 

For case (2), we have
\begin{align*}
&\frac{\partial^2 \infdiv*{\pmb\Gamma_1}{\pmb\Gamma_2}}{\partial b_1^2}= \frac{4}{\left(4 b_1^2-1\right)};~\frac{\partial^2 \infdiv*{\pmb\Gamma_1}{\pmb\Gamma_2}}{\partial b_1 \partial b_2}= 0;~\frac{\partial^2 \infdiv*{\pmb\Gamma_1}{\pmb\Gamma_2}}{\partial b_1 \partial \theta}= \frac{2}{4b_1^2-1} \sqrt{4\left( \tilde{c}\cos\delta\phi+\tilde{\upsilon}\sin\delta\phi \right)^2+\tilde{\beta}^2};\\
&\frac{\partial^2 \infdiv*{\pmb\Gamma_1}{\pmb\Gamma_2}}{\partial b_1 \partial \delta\phi}=-\frac{2\tilde{\beta}\left( - \tilde{c}\sin\delta\phi+\tilde{\upsilon}\cos\delta\phi \right) }{(1-4 b_1^2)\sqrt{4\left( \tilde{c}\cos\delta\phi+\tilde{\upsilon}\sin\delta\phi \right)^2+\tilde{\beta}^2}};~\frac{\partial^2 \infdiv*{\pmb\Gamma_1}{\pmb\Gamma_2}}{\partial b_2^2}= \frac{4}{\left(4 b_1^2-1\right)};\\
&\frac{\partial^2 \infdiv*{\pmb\Gamma_1}{\pmb\Gamma_2}}{\partial b_2 \partial \theta}= -\frac{2}{4b_1^2-1} \sqrt{4\left( \tilde{c}\cos\delta\phi+\tilde{\upsilon}\sin\delta\phi \right)^2+\tilde{\beta}^2};~\frac{\partial^2 \infdiv*{\pmb\Gamma_1}{\pmb\Gamma_2}}{\partial b_2 \partial \delta\phi}=  \frac{2\tilde{\beta}\left( - \tilde{c}\sin\delta\phi+\tilde{\upsilon}\cos\delta\phi \right) }{(1-4 b_1^2)\sqrt{4\left( \tilde{c}\cos\delta\phi+\tilde{\upsilon}\sin\delta\phi \right)^2+\tilde{\beta}^2}};\\
&\frac{\partial^2 \infdiv*{\pmb\Gamma_1}{\pmb\Gamma_2}}{\partial \theta^2}= 0;~\frac{\partial^2 \infdiv*{\pmb\Gamma_1}{\pmb\Gamma_2}}{\partial \theta\partial\delta\phi}=\frac{2(a_2-a_1) \left( \tilde{c}\cos\delta\phi+\tilde{\upsilon}\sin\delta\phi \right) \left( - \tilde{c}\sin\delta\phi+\tilde{\upsilon}\cos\delta\phi \right)}{\sqrt{4\left( \tilde{c}\cos\delta\phi+\tilde{\upsilon}\sin\delta\phi \right)^2+\tilde{\beta}^2}};\\
&\frac{\partial^2 \infdiv*{\pmb\Gamma_1}{\pmb\Gamma_2}}{\partial\delta\phi^2}=\frac{\tilde{\beta} (a_2-a_1) \left(  \tilde{c}\cos\delta\phi+\tilde{\upsilon}\sin\delta\phi \right)}{2 \sqrt{4\left( \tilde{c}\cos\delta\phi+\tilde{\upsilon}\sin\delta\phi \right)^2+\tilde{\beta}^2} }.
\end{align*} 
\end{widetext}
Since in this case the angle $\delta\phi$ doesnot affect the minimization, we can choose it as per our convenience. Without any loss of generality, we choose $\tan\delta\phi=-\tilde{c}/\tilde{\upsilon}$, which implies $\cos\delta\phi=\tilde{\upsilon}/\sqrt{\tilde{c}^2+\tilde{\upsilon}^2}$ and $\sin\delta\phi=-\tilde{c}/\sqrt{\tilde{c}^2+\tilde{\upsilon}^2}$. The Hessian matrix in this case becomes
 \begin{align*}
 H^* =\frac{1}{4b_1^2-1}
\begin{pmatrix}
4 & 0 &  2\tilde{\beta} & -2\tilde{w}\\
0 &4 & -2\tilde{\beta} & 2\tilde{w}\\
2\tilde{\beta} &-2\tilde{\beta} & 0 & 0\\
-2\tilde{w}& 2\tilde{w} & 0 & 0
\end{pmatrix},
\end{align*}
where $\tilde{w}=\sqrt{\tilde{c}^2+\tilde{\upsilon}^2}$.
This matrix has eigenvalues $\frac{2}{4b_1^2-1}\{2,0,1-\sqrt{1+2\tilde{\beta}^2+2\tilde{w}^2}, 1+\sqrt{1+2\tilde{\beta}^2+2\tilde{w}^2}\}$. Since there exists a negative eigenvalue of the Hessian, therefore, this case doesnot correspond to a minimum.

For case 3, we have $\frac{\partial^2 \infdiv*{\pmb\Gamma_1}{\pmb\Gamma_2}}{\partial b_1^2}= \frac{4}{\left(4 b_1^2-1\right)}$, $\frac{\partial^2 \infdiv*{\pmb\Gamma_1}{\pmb\Gamma_2}}{\partial b_1 \partial b_2}= 0$, $\frac{\partial^2 \infdiv*{\pmb\Gamma_1}{\pmb\Gamma_2}}{\partial b_1 \partial \theta}= 0$, $\frac{\partial^2 \infdiv*{\pmb\Gamma_1}{\pmb\Gamma_2}}{\partial b_1 \partial \delta\phi}= 0$, $\frac{\partial^2 \infdiv*{\pmb\Gamma_1}{\pmb\Gamma_2}}{\partial b_2^2}= \frac{1}{\left(4 b_2^2-1\right)}$, $\frac{\partial^2 \infdiv*{\pmb\Gamma_1}{\pmb\Gamma_2}}{\partial b_2 \partial \theta}= 0$, $\frac{\partial^2 \infdiv*{\pmb\Gamma_1}{\pmb\Gamma_2}}{\partial b_2 \partial \delta\phi}=  0$, $\frac{\partial^2 \infdiv*{\pmb\Gamma_1}{\pmb\Gamma_2}}{\partial \theta^2}= (a_2-a_1)\tilde{\beta}$, $\frac{\partial^2 \infdiv*{\pmb\Gamma_1}{\pmb\Gamma_2}}{\partial \theta\partial\delta\phi}=(a_2-a_1)\sqrt{\tilde{c}^2+\tilde{\upsilon}^2}$, and $\frac{\partial^2 \infdiv*{\pmb\Gamma_1}{\pmb\Gamma_2}}{\partial\delta\phi^2}=0$.
 The eigenvalues of the Hessian matrix in this case are $\left\{\frac{4}{\left(4 b_1^2-1\right)}, \frac{4}{\left(4 b_2^2-1\right)}, \frac{(a_2-a_1)}{2}\left( \tilde{\beta}\pm \sqrt{\tilde{\beta}^2+4(\tilde{c}^2+\tilde{\upsilon}^2)}\right)\right\}$. So, one eigenvalue is always negative in this case. Therefore, this case also doesnot correspond to a minimum. Thus, the above analysis shows that the only minimum of relative entropy $\infdiv*{\pmb\Gamma_1}{\pmb\Gamma_2}$ is given by case (1). Therefore, for a Gaussian state $\rho(\pmb\Gamma,\mathbf{d})$, the relative entropy of local activity is given by
\begin{align*}
\act(\rho)=\sum_{i=1}^2 \left[g(b_i^*)-g(\nu_i)\right],
\end{align*}
where $\nu_i$ $(i=1,2)$ are symplectic eigenvalues of covariance matrix $\pmb\Gamma$ of $\rho$, and $b_i^*$ $(i=1,2)$  are given by Eq. \ref{eq:app-b1b2}.

\bibliography{loc-ath}
\end{document}